\documentclass[sigconf,natbib=false,screen,nonacm]{acmart}

\pdfoutput=1
\makeatletter
\def\@parfont{\bfseries}
\makeatother

\usepackage{xifthen}
\usepackage[utf8]{inputenc}
\usepackage[T1]{fontenc}
\makeatletter
\newcommand\ifBeamerThenElse[2]{\@ifclassloaded{beamer}{#1}{#2}}
\makeatother

\newcommand{\usepackageIfBeamer}[2][]{\ifthenelse{\isempty{#1}}{\ifBeamerThenElse{\usepackage{#2}}{}}{\ifBeamerThenElse{\usepackage[#1]{#2}}{}}}
\newcommand{\usepackageIfNotBeamer}[2][]{\ifthenelse{\isempty{#1}}{\ifBeamerThenElse{}{\usepackage{#2}}}{\ifBeamerThenElse{}{\usepackage[#1]{#2}}}}

\usepackage{pifont}
\usepackage{import}
\usepackage{pdfpages}

\usepackage{amsmath}
\usepackage{amsthm}

\usepackage{pifont}
\usepackage{array}
\usepackage{siunitx}
\usepackage[operators, sets, lambda]{cryptocode}

\usepackage{graphicx}

\usepackage{tikz}
\usetikzlibrary{shapes,arrows,trees,matrix,positioning,decorations.pathreplacing,calc}
\usepackage{tcolorbox}

\usepackage{subcaption}

\usepackage{tabularx}
\usepackage{booktabs}

\usepackage{listings}

\usepackage{hyperref}
\hypersetup{
    colorlinks=true,
    allcolors=black
}
\usepackage[htt]{hyphenat}
\usepackage{csquotes}
\usepackage{enumerate}
\usepackage[normalem]{ulem}
\usepackage{todonotes}
\usepackage{enumitem}
\usepackage{acronym}
\usepackage{microtype}
\usepackage{caption}
\RequirePackage[
    datamodel=acmdatamodel,
    style=acmnumeric,
]{biblatex}
\usepackage{xurl}
\usepackage[capitalise, noabbrev]{cleveref}
\pcfixcleveref

\setlist[itemize]{nosep,leftmargin=*}
\setlist[enumerate]{nosep,leftmargin=*}

\ifBeamerThenElse{}{\theoremstyle{definition}
    \newtheorem{defn}{Definition}
    \crefname{defn}{Definition}{Definitions}

    \theoremstyle{plain}
    \newtheorem{prop}{Proposition}
    \crefname{prop}{Proposition}{Propositions}

    \theoremstyle{definition}
    \newtheorem{rem}{Remark}
    \crefname{rem}{Remark}{Remarks}
}

\DeclareSourcemap{
    \maps[datatype=bibtex,overwrite]{
        \map{
            \step[fieldset=pages, null]
            \step[fieldset=volume, null]
            \step[fieldset=number, null]
            \step[fieldset=doi, null]
            \step[fieldset=url, null]
            \step[fieldset=urldate, null]
            \step[fieldset=issn, null]
            \step[fieldset=isbn, null]
            \step[fieldset=publisher, null]
            \step[fieldset=address, null]
            \step[fieldset=location, null]
            \step[fieldset=editor, null]
            \step[fieldset=series, null]
            \step[fieldset=month, null]
            \step[fieldset=bibsource, null]
            \step[fieldset=biburl, null]
            \step[fieldset=eventdate, null]
            \step[fieldset=eventlocation, null]
            \step[fieldset=timestamp, null]
        }\map{
            \step[fieldsource=doi,
                match={\\_},
            replace={_}]
            \step[fieldsource=url,
                match={\\_},
            replace={_}]
        }
        \map{\step[fieldset=booktitle, fieldsource=booktitle,
                match=\regexp{Conference\s+on\s+Computer\s+and\s+Communications\s+Security},
            final=true]
            \step[fieldset=booktitle, fieldvalue=ACM CCS, origfieldval=false]
        }
        \map{
            \step[fieldset=booktitle, fieldsource=booktitle,
                match=\regexp{IEEE\s+Symposium\s+on\s+Security\s+and\s+Privacy},
            final=true]
            \step[fieldset=booktitle, fieldvalue={IEEE S\&P}, origfieldval=false]
        }
        \map{
            \step[fieldset=booktitle, fieldsource=booktitle,
                match=\regexp{USENIX\s+Security},
            final=true]
            \step[fieldset=booktitle, fieldvalue=USENIX Security, origfieldval=false]
        }
        \map{
            \step[fieldset=booktitle, fieldsource=booktitle,
                match=\regexp{Network\s+and\s+Distributed\s+System\s+Security},
            final=true]
            \step[fieldset=booktitle, fieldvalue=ISOC NDSS, origfieldval=false]
        }
        \map{
            \step[fieldset=booktitle, fieldsource=booktitle,
                match=\regexp{EUROCRYPT},
            final=true]
            \step[fieldset=booktitle, fieldvalue=IACR EUROCRYPT, origfieldval=false]
        }
        \map{
            \step[fieldset=booktitle, fieldsource=booktitle,
                match=\regexp{ASIACRYPT},
            final=true]
            \step[fieldset=booktitle, fieldvalue=IACR ASIACRYPT, origfieldval=false]
        }
        \map{
            \step[fieldset=booktitle, fieldsource=booktitle,
                match=\regexp{CRYPTO},
            final=true]
            \step[fieldset=booktitle, fieldvalue=IACR CRYPTO, origfieldval=false]
        }
        \map{
            \step[fieldset=booktitle, fieldsource=booktitle,
                match=\regexp{LATINCRYPT},
            final=true]
            \step[fieldset=booktitle, fieldvalue=IACR LATINCRYPT, origfieldval=false]
        }
        \map{
            \step[fieldset=booktitle, fieldsource=booktitle,
                match=\regexp{Public-Key\s+Cryptography},
            final=true]
            \step[fieldset=booktitle, fieldvalue=IACR PKC, origfieldval=false]
        }
        \map{
            \step[fieldset=booktitle, fieldsource=booktitle,
                match=\regexp{Public-Key\s+Cryptography},
            final=true]
            \step[fieldset=booktitle, fieldvalue=IACR PKC, origfieldval=false]
        }
        \map{
            \step[fieldset=booktitle, fieldsource=booktitle,
                match=\regexp{Public-Key\s+Cryptography},
            final=true]
            \step[fieldset=booktitle, fieldvalue=IACR PKC, origfieldval=false]
        }
        \map{
            \step[fieldset=booktitle, fieldsource=booktitle,
                match=\regexp{Public\s+Key\s+Cryptography},
            final=true]
            \step[fieldset=booktitle, fieldvalue=IACR PKC, origfieldval=false]
        }
        \map{
            \step[fieldset=booktitle, fieldsource=booktitle,
                match=\regexp{Security\s+and\s+Cryptography\s+for\s+Networks},
            final=true]
            \step[fieldset=booktitle, fieldvalue=SCN, origfieldval=false]
        }
        \map{
            \step[fieldset=booktitle, fieldsource=booktitle,
                match=\regexp{WTSC},
            final=true]
            \step[fieldset=booktitle, fieldvalue=WTSC, origfieldval=false]
        }
        \map{
            \step[fieldset=journal, fieldsource=journal,
                match=\regexp{Proc.\s+Priv.\s+Enhancing\s+Technol.},
            final=true]
            \step[fieldset=journal, fieldvalue=PoPETS, origfieldval=false]
        }
        \map{
            \step[fieldset=booktitle, fieldsource=booktitle,
                match=\regexp{Financial\s+Cryptography},
            final=true]
            \step[fieldset=booktitle, fieldvalue=FC, origfieldval=false]
        }
    }
}

\newcommand{\subheading}[1]{\vspace{0.5 em}\noindent \textbf{#1}}

\newcommand*{\emptyCircle}{
    \begin{tikzpicture}[scale=0.11]\draw (0,0) circle (1);
\end{tikzpicture}}
\newcommand*{\halfCircle}{
    \begin{tikzpicture}[scale=0.11]\draw (0,0) circle (1);
        \fill (0,0) -- (90:1) arc (90:90-50*3.6:1) -- cycle;
\end{tikzpicture}}
\newcommand*{\fullCircle}{
    \begin{tikzpicture}[scale=0.11]\draw (0,0) circle (1);
        \fill (0,0) -- (90:1) arc (90:90-100*3.6:1) -- cycle;
\end{tikzpicture}}

\newcommand\cmark{\fullCircle}
\newcommand\xmark{\emptyCircle}

\newcommand\specialfont[1]{{\texttt{#1}}}

\newcommand{\belowLine}{4pt}

\acrodef{prim}[PVC]{Permissioned Vector Commitment}
\acrodef{model}[PPoL]{Permissioned Proof of Liabilities}

\renewcommand\Pr{\mathrm{Pr}}

\newcommand\pccheck{\text{\textbf{check} }}
\newcommand\pcComment[1]{\quad\ \text{// #1}}
\newcommand\pcCommentLine[1]{\text{// #1}}

\newcommand\domain{H}
\newcommand\subgroup{\domain}
\newcommand\lag{\ell}
\newcommand\perm{\alpha}
\newcommand\hash{\specialfont{H}}
\newcommand\groupHash{\hash}
\newcommand\fieldHash{\hash'}
\newcommand\sig{\sigma}
\newcommand\sk{\specialfont{sk}}
\newcommand\pk{\specialfont{pk}}
\newcommand\apk{\specialfont{apk}}
\newcommand\dbPoly{v}
\newcommand\dbCom{V}
\newcommand\keyPoly{s}
\newcommand\keyCom{S}
\newcommand\dbProof{\pi}
\newcommand\keyProof{\varphi}
\newcommand\aux{\specialfont{mem}}

\newcommand\sigPoly{f}
\newcommand\sigCom{F}
\newcommand\binPoly{b}
\newcommand\binCom{B}
\newcommand\zeroQuo{q}
\newcommand\zeroQuoCom{Q}
\newcommand\apkProof{\pi_{\apk}}
\newcommand\binQuo{u}
\newcommand\binQuoCom{U}
\newcommand\apkQuo{t}
\newcommand\apkQuoCom{T}
\newcommand\apkRem{r}
\newcommand\apkRemCom{R}
\newcommand\apkRemDeg{p}
\newcommand\apkRemDegCom{P}
\newcommand\lagQuo{L}
\newcommand\lagProof{\lagQuo}
\newcommand\newKeyProof{\psi}
\newcommand\total{z}
\newcommand\Total{Z}
\newcommand\totalCom{\Total}
\newcommand\maskPoly{w}
\newcommand\maskCom{W}
\newcommand\polQuo{q'}
\newcommand\polQuoCom{Q'}
\newcommand\sumProof{\pi_\mathrm{sum}}
\newcommand\polBin{d}
\newcommand\polBinCom{D}
\newcommand\rangeProof{\pi_\mathrm{range}}

\newcommand\rangeQuoCom{E}
\newcommand\RangeProve{\specialfont{Range.Prove}}
\newcommand\RangeVerify{\specialfont{Range.Verify}}

\newcommand\KeyGen{\specialfont{KeyGen}}
\newcommand\Setup{\specialfont{Setup}}
\newcommand\Sign{\specialfont{Sign}}

\newcommand\VerifyLookup{\specialfont{VerifyLU}}
\newcommand\VerifyLU{\VerifyLookup}

\newcommand\Publish{\specialfont{EndEpoch}}
\newcommand\CloseEpoch{\Publish}
\newcommand\EndEpoch{\Publish}
\newcommand\VerifyPK{\specialfont{VerifyPK}}
\newcommand\UpdatePK{\specialfont{AddPK}}
\newcommand\AddPK{\UpdatePK}
\newcommand\VerifySig{\specialfont{VerifySig}}
\newcommand\UpdateDB{\specialfont{UpdateDB}}

\newcommand\VerifyDB{\specialfont{VerifyDB}}
\newcommand\VerifyKeys{\specialfont{VerifyKeys}}
\newcommand\VerifyHints{\specialfont{VerifyHints}}
\newcommand\VerifyHelpers{\VerifyHints}
\newcommand\VerifyAppend{\specialfont{VerifyAppend}}
\newcommand\APKVerify{\specialfont{APK.Verify}}
\newcommand\primUpdateDB{\specialfont{PVC.UpdateDB}}
\newcommand\primEndEpoch{\specialfont{PVC.EndEpoch}}
\newcommand\primVerifyDB{\specialfont{PVC.VerifyDB}}

\newcommand\pp{\specialfont{pp}}
\newcommand\KZGSetup{\specialfont{AMT.Setup}}

\newcommand\KZGCommit{\specialfont{AMT.Commit}}

\newcommand\AMTGetProof{\specialfont{AMT.Open}}
\newcommand\AMTOpen{\AMTGetProof}

\newcommand\AMTVerify{\specialfont{AMT.Verify}}
\newcommand\AMTMaintain{\specialfont{AMT.Update}}

\newcommand\gameOne{\mathcal{E}_1}
\newcommand\gameTwo{\mathcal{E}_2}
\newcommand\gameThree{\mathcal{E}_3}
\newcommand\gameFour{\mathcal{E}_4}
\newcommand\gameFive{\mathcal{E}_5}
\newcommand\gameDbSec{\gameOne}
\newcommand\gameZkReal{\gameTwo}
\newcommand\gameZkSim{\gameThree}
\newcommand\gameUf{\gameFour}
\newcommand\gameSingleRound{\gameFive}
\newcommand\advA{\mathcal{A}}
\newcommand\advB{\mathcal{B}}
\newcommand\Sim{\mathcal{S}}
\newcommand\expList{\mathcal{L}}
\begin{document}

\title{Mitigating Collusion in Proofs of Liabilities}

\author{Malcom Mohamed}
\affiliation{\institution{Ruhr University Bochum}
    \city{Bochum}
    \country{Germany}
}
\email{malcom.mohamed@rub.de}
\orcid{0000-0002-9865-1244}

\author{Ghassan Karame}
\affiliation{\institution{Ruhr University Bochum}
    \city{Bochum}
    \country{Germany}
}
\email{ghassan@karame.org}
\orcid{0000-0002-2828-4071}

\begin{abstract}
    Cryptocurrency exchanges use proofs of liabilities (PoLs) to prove to their customers their liabilities committed on-chain, thereby enhancing their trust in the service.
    Unfortunately, a close examination of currently deployed and academic PoLs reveals significant shortcomings in their designs.
    For instance, existing schemes cannot resist realistic attack scenarios in which the provider colludes with an existing user.

    In this paper, we propose a new model, dubbed permissioned PoL, that addresses this gap by not requiring cooperation from users to detect a dishonest provider's potential misbehavior.
    At the core of our proposal lies a novel primitive, which we call Permissioned Vector Commitment (PVC), to ensure that a committed vector only contains values that users have explicitly signed.
    We provide an efficient PVC and PoL construction that carefully combines homomorphic properties of KZG commitments and BLS-based signatures.
    Our prototype implementation shows that, despite the stronger security, our proposal also improves server performance (by up to $10\times$) compared to prior PoLs.
\end{abstract}

\maketitle
\settopmatter{printfolios=true}
\settopmatter{printacmref=false}

\section{Introduction}\label{sec:intro}
Cryptocurrency users commonly rely on centralized service providers (centralized exchanges, CEXs) to convert between cryptocurrencies or avoid the burden of keeping valuable private keys in personal custody.
In recent years, numerous accounts have highlighted instances where CEXs have abruptly declared bankruptcy, leaving their users without any prior warning \cite{coindeskFTX}.
To enhance their trustworthiness, major CEXs, such as Binance~\cite{binanceImpl}, Kraken~\cite{krakenImpl}, OKX~\cite{okxImpl} and Gate.io~\cite{gateioImpl}, currently offer mechanisms intended to guarantee the availability and liquidity of their stored funds, referred to as \emph{proofs of solvency} or \emph{proofs of reserve}.
These proofs aim to enable decentralized solvency audits in the sense that the collective of users takes up a role traditionally reserved for private third-party auditors.
Currently, these proofs of solvency consist of the CEXs
\textbf{(1)} publishing a cryptographic commitment of all customers' account balances, where users may check their individual balances' inclusion,
\textbf{(2)} proving that the sum of balances committed is at most a given value $X$ and \textbf{(3)} proving that they own assets worth at least $X$ units.
Steps 1 and 2 combined are also referred to as \emph{proof of liabilities (PoL)} and are the focus of this paper.
Given a PoL, step 3---also called \emph{proof of assets}---can be instantiated as a sign-off from a trusted auditor, the publication of a list of the exchange's owned blockchain addresses and/or zero-knowledge techniques \cite{gotzilla}.
The outcome of proofs of solvency are statements like \enquote{$p\%$ of deposits are fully backed} (usually, $p\geq100$) that are regularly published.

Unfortunately, an examination of currently deployed and academic PoL designs \cite{maxwell,dapolplus,notus}, reveals several shortcomings that prevent their large-scale adoption in real-world deployments.
First and foremost, state-of-the-art PoLs, such as Notus~\cite{notus}, are designed to ensure that a malicious service provider cannot underreport the aggregate liabilities of their users.
However, they cannot guarantee that this aggregate truly reflects the current balances of \emph{all} users.
The only means in existing PoLs to ensure that the balances are accurately reflected is for every user to independently verify the inclusion of their own balances in the proof---a requirement that burdens users and compels providers to adopt long reporting intervals to keep the process reasonably user-friendly.

Beyond its obvious inconvenience for users, such a design is inherently flawed since \emph{it cannot deter collusion between users and the service provider}. Consider the realistic example where Alice and Bob both use an exchange attempting to conceal its insolvency. Just before the next PoL audit, the provider could coerce Alice into colluding without risk of detection. Namely, the provider would produce a proof with a reduced balance for Alice (e.g., by setting the balance of Alice to 0), thereby reducing the reported total liabilities to a level it can cover. In this scenario, such a collusion can invalidate the entire PoL. In return, Alice could receive her full balance, perhaps with additional benefits/interests, immediately after the reporting period.

We argue that this limitation significantly hinders the usefulness of current PoLs, as CEXs have hundreds of thousands of users.
One should not expect that \emph{all} those users do not deviate from their honest behavior in the process. Moreover, colluding users can even retain plausible deniability by claiming that they forgot to perform the check and were merely victims of the malicious provider.
Another important limitation of existing proposals lies in their computational cost, with some requiring a few minutes to create a proof. For instance, proving the liabilities of 2000 users takes around 180 seconds in the state-of-the-art proposal by~\cite{notus}.

In this work, we address these gaps and propose a fundamentally new PoL model, dubbed permissioned PoL (\acs{model}), that addresses these limitations by not requiring cooperation from users to detect potential misbehavior by a malicious provider.
Our model does {not} allow the provider, in principle, to tamper with users' balance values and does {not} burden users with constantly performing checks to detect potential tampering {after the fact}.
Instead, our model enforces a permission policy by which tampering is prevented {a priori}.
At the core of our proposal lies a novel abstraction, which we dub \acf{prim}, that ensures that a vector commitment only contains values the users have signed.
The \acs{prim} formalism captures precisely the PKI integration required for a \acs{model}.
A \acs{prim} is beneficial in situations like the proof of liabilities setting, where a centralized provider performs a global computation over an outsourced database and outputs the result and a proof of correctness.
We show that this functionality can be efficiently instantiated and that an efficient \acs{model} can be built on top of such an instantiation.
The proposed solution combines recent ideas (like AMTs \cite{tomescu20} and APK proofs \cite{das23}) with a novel quotient-based signature aggregation technique and an append proof for committed secrets. We implement a prototype with code publicly available at \cite{githubRepo}.
Our evaluation shows that our scheme---despite its stronger security---scales better than prior work (by up to $10\times$).

\subheading{Contributions.}
In summary, we make the following contributions:
\begin{description}
    \item[Gaps in Existing PoL solutions: ]
        We analyze currently commercialized and academic PoL designs and show that they suffer from significant security limitations that hinder their usefulness for cryptocurrency users (\cref{sec:liabilities}).
        We then propose a new PoL notion, the \acf{model}, that aims to overcome these limitations.
    \item[\acf{prim}: ]
        We introduce the \acf{prim} as an abstraction that cleanly captures precisely the PKI integration required for a \ac{model} (\cref{sec:prim}).
        We argue that the \ac{prim} is a meaningful independent abstraction that also strengthens other transparency systems besides PoLs.
    \item[Efficient \ac{prim}: ]
        We construct an efficient \ac{prim} using KZG commitments, carefully combining recent ideas (like AMTs and APK proofs) with a novel quotient-based signature aggregation technique and an append proof for committed secrets (\cref{sec:prelims,sec:primConstruct}).
    \item[Design \& Implementation of an Efficient \ac{model}: ]
        We design and implement a \ac{prim}-based \ac{model}.
        Our evaluation (\cref{sec:pol}) shows that our scheme---despite its stronger security---scales better than prior work (by up to $10\times$).
\end{description}

\section{Proofs of Liabilities}\label{sec:survey}\label{sec:related}\label{sec:liabilities}

\subsection{System Model}\label{sec:basicPoL}

Proof of Liabilities (PoL) systems involve three roles (and an implicit \emph{Judge}): users, a provider and auditors \cite{provisions,notus}, with anyone able to act as an auditor.
The provider maintains broadcast communication with users and auditors and private channels with individual users.
It manages a database mapping each user to their account balance, which remains hidden from other users and auditors.

The system operates in epochs.
During each epoch, users can contact the provider to request updates to their respective account balances.
At the end of an epoch, the provider publishes a cryptographic commitment to the database and (a commitment to) the sum of all users' balances.
The standard threat model of PoL has trusted users, trusted auditors and an untrusted provider---though in this paper, we will also consider dishonest users.
A malicious provider may attempt to output a (commitment to a) sum of liabilities that is \emph{lower} than the actual sum of users' balances.
To prevent this, each newly published database commitment must pass two kinds of checks.
First, users request inclusion proofs for their individual balances.
Second, auditors verify that the committed total matches the database.
In some designs \cite{maxwell,falzon23}, this check requires users' private balances, forcing them to perform it themselves and removing the auditor role.

An underlying assumption is that users and auditors have a \emph{Judge} to complain to if their local checks detect misbehavior on part of the provider.
Sufficiently many cases of
users detecting a wrong balance at their database entry or of
auditor verification failing would trigger punitive action against the provider.

\subsection{Existing Schemes}\label{sec:existing}

\subheading{Summation trees.}
In Maxwell and Todd's design \cite{maxwell}, the database commitment is a special Merkle tree in which leaves stores account balances and each inner node stores not only the hash of its two children but also the sum of the children's associated balance values.
The root node stores the sum of all balances.
Then, each user's inclusion proof contains all the intermediate sums stored in the tree's nodes.
Any malicious underreporting of total liabilities would mean a user's balance was not included, which would be discovered when that user asked for an inclusion proof.
The scheme has no auditor role.
Each inclusion proof leaks partial information about other users' balances.

We describe three works from the literature that improve on Maxwell and Todd's design \cite{maxwell} (\cref{sec:existing}).
The Provisions scheme \cite{provisions} addresses the privacy leakage by replacing the Merkle tree with a list of hiding Pedersen commitments.
Instead of publishing the Merkle root and sum of balances, the entire list of commitments is published.
Each user $i$ checks that the $i$th commitment contains the user's balance.
Then, anyone can homomorphically compute a privacy-preserving commitment to the sum of all balances.
The DAPOL+ scheme \cite{dapolplus} combines hiding commitments with the Merkle tree approach in order to reduce the published data back down to a single Merkle root (instead of a list of commitments like Provisions).
The scheme from Falzon et al. \cite{falzon23} reduces the size of inclusion proofs by replacing the DAPOL+ Merkle tree with higher-arity vector commitment trees combined with range proofs and inner-product arguments.
We refer to \cite{chalkiasBroken} for a discussion of further schemes in this line of research.

\def\xDistance{2.4}
\newcommand{\otbAttack}[5]{
    \node (title) at (1.335*\xDistance,1.75) {#1};
    \node[align=left] (alice1) at (0,1) {Real liability\\to Alice: $a$};
    \node[align=left,below=of alice1] (bob1) {Real liability\\to Bob: $b$};
    \node[draw,align=left] (comp1) at (\xDistance,0.5) {Claimed\\total:\\#2};
    \node[draw,align=left] (verify1) at (1.95*\xDistance,0.5) {Honest users\\#3{}check};
    \node[align=left] (result1) at (2.67*\xDistance,0.5) {\textbf{\checkmark}};
    \node (comment1) at (1.0*\xDistance,-0.5) {#4};
    \node (comment2) at (2.45*\xDistance,-0.5) {#5};
    \draw[->] (alice1) -- (comp1);
    \draw[->] (bob1) -- (comp1);
    \draw[->] (comp1) -- (verify1);
    \draw[->] (verify1) -- (result1);
}
\begin{figure*}[!t]
    \centering
    \begin{tikzpicture}[node distance=0.2cm and 1.0cm]
        \otbAttack{\textbf{Epoch 1:} Alice and Bob are honest.}{$a+b$}{}{\phantom{abcd}}{}
    \end{tikzpicture}
    \vrule
    \begin{tikzpicture}[node distance=0.2cm and 1.0cm]
        \otbAttack{\textbf{Epoch 2:} Alice makes an \emph{off-the-books deal}.}{$b$}{(Bob) }{($< a+b$)}{(\textbf{Attack})}
    \end{tikzpicture}
    \caption{OTB attack, undetectable by prior PoLs.
        Meanwhile, Alice is risk-free.
    \label{fig:otb}}
\end{figure*}
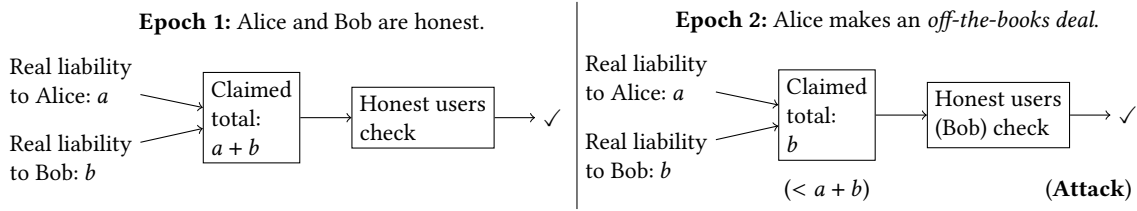

\subheading{SNARKs.}
Next, we describe the PoL that is deployed in production by Binance \cite{binanceImpl}.
Other real-world deployments can be seen as adaptations of the Binance scheme \cite{okxImpl,gateioImpl,krakenImpl}.
The Binance design differs from the previously described proposals because it is based on general-purpose SNARKs instead of summation trees.
Specifically, Binance uses a standard Merkle tree for the user balances.
Users' individual inclusion checks, therefore, do not implicitly (partially) verify the sum of liabilities like in the Maxwell-Todd-based designs.
Instead, the provider creates a zkSNARK that the elements in the Merkle tree leaves add up to the claimed (committed) sum value.\footnote{The SNARK also has other features that are not relevant for this paper, e.g., it supports liabilities in multiple currencies.} With zkSNARKs, it is possible to let an auditor succinctly verify the summation without any private information.
This is in contrast to the previously discussed summation-tree schemes.
In those, to verify that the summation is carried out correctly, \emph{all users' proofs} would have to be re-checked by the auditor.
However, the SNARK adds a non-negligible overhead for the provider as, reportedly, a 32-core machine takes around 1 minute to create a proof for 700 users.

A recent preprint \cite{xiezhi} presents Xiezhi, a SNARK-based PoL design with independent publicly verifiable summation proofs. Unlike Binance’s scheme, it avoids Merkle trees and general-purpose SNARKs, relying instead on specialized polynomial commitment techniques.

\subheading{Notus.}
A recent advancement is due to Xin et al., who introduced Notus \cite{notus}.
Notus departs from the prior designs in three significant ways.
First, the core data structure is not a Merkle tree but an RSA accumulator. Second, the data structure not only stores users' balances but also hash chain commitments to their entire \emph{transaction histories}.
Third, the provider also publishes zero-knowledge proofs that each new accumulator is an update of the previous accumulator such that the hash chain is correctly maintained at each updated entry.
Xin et al. term this architecture \emph{dynamic proof of liabilities (DPoL)} since it focuses on epoch-to-epoch updates rather than a static database.
The DPoL design guarantees an append-only property such that a tampering provider's malicious updates would necessarily be detected \emph{whenever} the affected user next performs a check.
Thus, users do not need to perform checks for every epoch.

\subsection{Limitations}\label{subsec:limitations}

As noted in the introduction, prior deployments of proofs of solvency typically used long epoch lengths (e.g., one month). This creates a trivial attack vector: a malicious exchange can borrow assets shortly before an audit to appear solvent and pass all checks. This weakness is well known and often cited in criticism \cite{coppola}.
A direct mitigation is to increase the frequency of proofs. However, existing PoL schemes then run into scalability issues because they require users to take action for every proof.
This may seem counterintuitive, since earlier works argued that only a small fraction of users (e.g., 0.05\%) need to participate for the system to remain secure. But that argument assumes that users’ checking behavior is independent and random \cite[Section 5]{dapolplus}.

Unfortunately, this argument breaks down for two reasons. First, as shown in \cite{statAttacks}, a malicious provider could try to predict which users are likely to check and tailor its tampering strategy accordingly, undermining the probabilistic model. Second, the analysis overlooks the possibility of malicious users colluding with the provider, whose checking probability is effectively zero.

This ``oversight'' stemmed from an unrealistic assumption that is baked into the prevailing PoL model.
Namely, it is generally assumed that the provider only needs to be stopped from under-reporting its liabilities to \emph{honest} users. It is explicitly left out of the threat model that \emph{dishonest} users could still represent \emph{real liabilities that ought to be taken into account when assessing the provider's solvency} \cite[Definition 3.3]{dapolplus} \cite[page 5]{notus}.
We argue that it is highly probable---and in fact, \emph{rational} behavior---for users to break this assumption.
To this end, we describe in what follows an off-the-books (OTB) attack that exploits this gap in the security model and is entirely risk-free for malicious users.

\subheading{OTB attacks.}
Consider the case where users Alice and Bob have non-zero account balances at the same exchange that is attempting to conceal its insolvency.
Shortly before the next PoL audit, Alice secretly colludes with the provider as follows:
The provider would use a value of 0 instead of Alice's true account balance, thereby reducing the stated liabilities to a level it appears able to cover. In return, Alice agrees not to verify her account’s inclusion in the liability audit and is granted withdrawal priority over Bob, along with additional benefits the provider may offer.
Through this arrangement, Alice helps create a false impression of solvency for the honest user Bob. Moreover, this is entirely one-sided: Alice can later back out if she suspects the provider may not be able to pay her after all. She can still perform her verification check at any time and complain to the \emph{Judge} that she was harmed by a malicious provider.
A rational Alice would therefore always accept such off-the-books deals.
\cref{fig:otb} illustrates this example.

\begin{table*}[t]
    \centering
    \caption{
        Comparison of existing PoLs.
        $n$ is the total number of users.
        We obtained the last row's data from \cref{fig:evalThroughput1}.
        \label{tab:related}
    }
    \vspace{-0.25cm}
    \begin{tabular}{l|l|l|l|l}
        \toprule
        & Standard PoLs \cite{maxwell,provisions,dapolplus,falzon23} & Succinct PoLs \cite{binanceImpl}, \cite{xiezhi} & Notus \cite{notus} &  PPoL \\
        \midrule
        OTB attacks non-repudiable & \xmark & \xmark & \xmark & \cmark \\
        No preemptive user checks (cf. \cref{rem:checks})
        & \xmark & \xmark & \halfCircle & \cmark \\
        Auditors' global verification & $O(n \cdot \log(n))$ & $O(1)$ & $O(1)$ & $O(\log^2(n))$ \\
        Per-user inclusion verification & $O(\log(n))$ & $O(1)$ (Xiezhi) & $O(1)$ & $O(\log(n))$ \\
        Updates/s by the server, $n=2^{16}$  & 0.267 (DAPOL+) & 145 (Xiezhi) & 3.69 & \textbf{369} \\
        \bottomrule
    \end{tabular}
\end{table*}

\subheading{Computational overhead.} When schemes do not use summation trees, the prover overhead can be substantial.
For illustration, consider that experiments from the state-of-the-art solution Notus indicate a 3-minute liability proving time on a multi-core machine for around $2^{11}$ updated balances \cite{notus}. By comparison, publicly available data show that, in a 3-minute window, large exchanges may face workloads of around 1 million balance updates (see \cref{fig:workload}).

\subsection{Towards a New Model}\label{sec:newModel}

We now derive requirements for our new PoL model---dubbed Permissioned PoL (PPoL)---and argue that it is best understood as a specialization of a minimal intermediary abstraction (introduced in \cref{sec:prim} as the \ac{prim}).
\cref{tab:related} compares the properties of our resulting PPoL construction (\cref{sec:primConstruct,sec:pol}) against prior systems.

\subheading{Desiderata for \acs{model}.}
We argue that a natural way of addressing OTB attacks is to introduce a notion of \emph{permissioning} in the following sense.
Specifically, we want auditors to be able to verify that every balance update in the database has been properly authorized by the corresponding user.
This verification must not rely on any private user data, allowing it to occur \emph{before} any user becomes active.
If a malicious Alice wants to allow the provider to use 0 instead of her true balance, she must make an explicit withdrawal transaction and her database entry must really become 0 in the published commitment.
This eliminates Alice's repudiation ability which allowed her to go back on an off-the-books deal and still convince the \emph{Judge} that she was a victim.
We thereby remove the incentive for users to \emph{always} agree to helping the provider cheat, addressing a major loophole in the PoL security model.

\begin{rem}[Minimum required checks]\label{rem:checks}
    Note that this design also removes the need for the burdensome preemptive checks of every published commitment on the part of the users, another notable improvement for PoLs.
    Indeed, the burden is now on the provider to present a publicly verifiable proof along with each new commitment---and the proof is verified by \emph{auditors}.
    Crucially, the proof now serves as a tamper prevention mechanism, rather than a detection mechanism. No prevention scheme guards against denial-of-service attacks in which a malicious provider would refuse a requested balance update.
    Thus, \ac{model} users still perform a \emph{one-time check after requesting an update} to ensure it has been correctly applied.
    \emph{This requirement only holds for users who issued an update during the epoch.}
\end{rem}

\begin{figure}[!t]
    \begin{center}
        \includegraphics[width=0.4\textwidth]{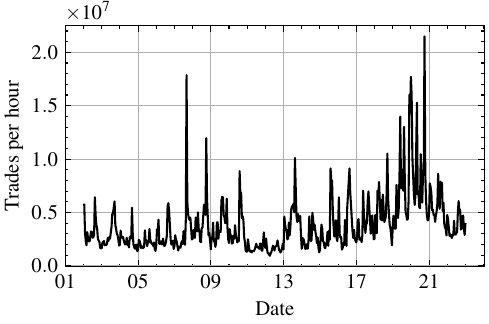}
        \vspace*{-0.5cm}
    \end{center}
    \caption{Executed USDT-denominated trades per hour on Binance between 2 January 2025 and 22 January 2025 \cite{binanceData}.}
    \label{fig:workload}
\end{figure}

\subheading{Determining the right abstraction level.}
To make the \ac{model} requirements described above concrete, we begin by considering the following key questions:
\emph{How can the system assign entry-wise write-access permission to users?
    How may users temporarily transfer the required permission to the provider?
Finally, how can an auditor globally check that all updates conform to the permission policy?}

Intuitively, any system addressing these questions must tightly integrate with an identity infrastructure or a PKI---at the very least, to identify authorized users.
Given the potential complexity of modeling the required (PKI) infrastructure's components and their interactions, we first try to reduce our scope.
As the reader will recall from \cref{sec:intro} and from \cref{sec:basicPoL}, a valid way to conceptualize PoLs is to view them as two interleaved subtasks:
\textbf{(1)} the construction and validation of a commitment to the database of user values and
\textbf{(2)} linking the database commitment to a (possibly also committed) claimed sum of liabilities.
This separation is crucial as we can observe that \emph{the permissioning only affects task \textbf{(1)}}.
If permission checks can be integrated into the construction of the database commitment, a modular \ac{model} is possible by layering an appropriate summation proof for task \textbf{(2)}.
This motivates an intermediate level of abstraction. In the sequel, we show how our \acl{prim} efficiently realizes exactly the required permission logic (see the last column of \cref{tab:related}).

\section{Permissioned Vector Commitment}\label{sec:prim}

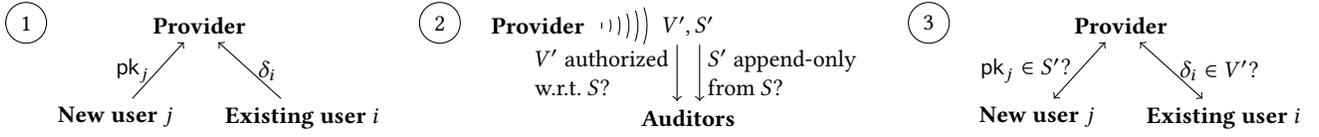
\begin{figure*}[!t]
    \centering
    \begin{tikzpicture}[
            node distance=0.75cm and -0.5cm,
        ]
        \node[circle,draw,anchor=north west] (num1) at (0,0) {1};
        \node[anchor=north] (mid1) at (2.5,0) {\textbf{Provider}};
        \node[below left=of mid1] (bot1a) {\textbf{New user $j$}};
        \node[below right=of mid1] (bot1b) {\textbf{Existing user $i$}};
        \draw[->] (bot1a) -- (mid1) node[left,midway] {$\pk_j$};
        \draw[->] (bot1b) -- (mid1) node[right,midway] {$\delta_i$};

        \node[circle,draw,anchor=north west] (num2) at (5.5,0) {2};
        \node[anchor=north] (mid2a) at (7,0) {\textbf{Provider}};
        \node[anchor=north] (mid2b) at (9,0) {$\dbCom',\keyCom'$};
        \node[below=of mid2b] (bot2) {\textbf{Auditors}};
        \draw[decoration={expanding waves, segment length=1.2mm, angle=20},decorate] (mid2a) -- (mid2b);
        \draw[<-] ([xshift=-0.15cm] bot2.north) -- ([xshift=-0.15cm] mid2b.south) node[left,midway,align=left] {$\dbCom'$ authorized\\w.r.t. $\keyCom$?};
        \draw[<-] ([xshift=0.15cm] bot2.north) -- ([xshift=0.15cm] mid2b.south) node[right,midway,align=left] {$\keyCom'$ append-only\\from $\keyCom$?};

        \node[circle,draw,anchor=north west] (num3) at (12,0) {3};
        \node[anchor=north] (mid3) at (14.75,0) {\textbf{Provider}};
        \node[below left=of mid3] (bot3a) {\textbf{New user $j$}};
        \node[below right=of mid3] (bot3b) {\textbf{Existing user $i$}};
        \draw[<->] (bot3a) -- (mid3) node[left,midway] {$\pk_j \in \keyCom'$?};
        \draw[<->] (bot3b) -- (mid3) node[right,midway] {$\delta_i \in \dbCom'$?};
    \end{tikzpicture}
    \vspace{-0.5 em}
    \caption{High-level \ac{prim} workflow.
        (1) Users make requests within an epoch.
        (2) At the end of the epoch, the provider publishes new commitments with publicly verifiable proofs.
        (3) The users who made requests confirm their updates' inclusion.
    \label{fig:setup}}
\end{figure*}

\subsection{Overview}\label{sec:primOverview}

A \ac{prim} is a multi-party system centered around two commitments $\dbCom$ and $\keyCom$.
$\keyCom$ is a commitment to the public keys of mutually distrustful \emph{users} and $\dbCom$ is a commitment to a database managed by an untrusted \emph{provider} where each user owns one entry.
At the end of each \emph{epoch}, the provider publishes new versions of $\dbCom$ and $\keyCom$. Each new version is accompanied by a global consistency proof, which is verified by trusted \emph{auditors}.
The provider has three primary tasks: registering a new user, updating a user's value in the database and publishing new versions of $\dbCom$ and $\keyCom$. Registered users perform the following operations: checking that their key is indeed included in $\keyCom$, signing an update to their entry in $\dbCom$ and checking that their entry in $\dbCom$ is correct.
After each publication, newly registered users verify the inclusion of their key in $\keyCom$ and users who requested an update to their values confirm that the new commitment includes their update.
Auditors use the global consistency proof to check that the provider
\textbf{(1)} only added new keys to $\keyCom$ but did not remove or change any existing key and
\textbf{(2)} did not tamper with $\dbCom$. That is, for each updated value in $\dbCom$, the corresponding user, as identified by their key committed in $\keyCom$, must have provided a valid signature over the update. The global consistency proof strengthens a \ac{prim}-based \ac{model} against OTB attacks by ensuring that \emph{every} updated entry is explicitly signed, making the attacker accountable.
Inactive users are \emph{not} required to check every new commitment, as any tampering with their value would be detected globally.\footnote{Inactive users are those who are registered in the system but who do not sign, or \enquote{refuse} to sign, a balance update during a given epoch. The system's progress does not depend on inactive users.}

\Cref{fig:setup} illustrates the workflows.

\subsection{Definitions}\label{sec:definitions}

We now give a formal definition of a \ac{prim} as a collection of algorithms for the various parties.
All algorithms can be stateful, i.e., each party's algorithms can read and write auxiliary data ($\aux$) that the party maintains in memory.
In the sequel, we may omit this state from our notation for ease of presentation.
\begin{defn}[\acl{prim}]\label{def:prim}
    A \acl{prim} (\acs{prim}) scheme consists of the following algorithms.
    For the provider:
    \begin{itemize}
        \item $\Setup(\secpar, n) \rightarrow \pp$.
            Given a security parameter $\secpar$ and an integer $n$, returns public parameters $\pp$ for up to $n$ users.
            $\pp$ is an implicit input to all algorithms.
        \item $\AddPK(i, \pk_i; \aux) \rightarrow \aux'$. Adds a new public key $\pk_i$ at index $i$ (and updates the internal state $\aux$ to $\aux'$).
        \item $\VerifySig(i, E, \pk_i, \delta_i, \sig_i; \aux) \rightarrow 0/1; \aux'$.
            Verifies the $i$th user's signature $\sig_i$ over an update $\delta_i$ for epoch $E$ (and updates the internal state $\aux$ to $\aux'$).
        \item $\UpdateDB(i, \delta_i, \sig_i; \aux) \rightarrow \aux'$. Applies an update $\delta_i$ to the data value at index $i$ (and updates $\aux$ to $\aux'$).
        \item $\Publish(\aux) \rightarrow \dbCom', \keyCom', \dbProof, \keyProof, \left(\dbProof_i\right)_{i\in I}, \left(\keyProof_i\right)_{i\in J}; \aux'$.
            Returns proofs $\dbProof$ and $\keyProof$ authenticating, respectively, the \emph{new} commitments $\dbCom'$ and $\keyCom'$. With $I,J \subseteq [0,n-1]$ being sets of indices, also returns inclusion proofs for all updated database values $i \in I$ and for all appended keys $i \in J$ (and updates the internal state $\aux$ to $\aux'$).
            \textbf{Note:} \EndEpoch{} returns \emph{all} required proofs.
    \end{itemize}

    \noindent{}For users:
    \begin{itemize}
        \item $\KeyGen(i) \rightarrow \sk_i, \pk_i$.
            Returns secret key $\sk_i$ and public key $\pk_i$.
        \item $\VerifyPK(\keyCom, i, \pk_i, \keyProof_i; \aux) \rightarrow 0/1; \aux'$.
            Verifies an inclusion proof $\keyProof_i$ for the $i$th key $\pk_i$ committed in $\keyCom$ (and updates the internal state $\aux$ to $\aux'$).
        \item $\Sign(\sk_i, E, \delta_i; \aux) \rightarrow \sig_i; \aux'$.
            Returns a signature $\sig_i$ for an update $\delta_i$ to the user's value in epoch $E$ (and updates the internal state to $\aux'$).
        \item $\VerifyLookup(\dbCom, i, v_i, \dbProof_i; \aux') \rightarrow 0/1; \aux'$.
            Verifies a lookup proof $\dbProof_i$ for the $i$th entry $v_i$ committed in $\dbCom$ (and updates the internal state $\aux$ to $\aux'$).
    \end{itemize}

    \noindent For auditors:
    \begin{itemize}
        \item $\VerifyDB(\keyCom, E, \dbCom', \dbProof; \aux) \rightarrow 0/1; \aux'$.
            Verifies the commitment $\dbCom'$ against $\keyCom$ using the proof $\dbProof$ (and updates the internal state $\aux$ to $\aux'$).
        \item $\VerifyKeys(\keyCom, E, \keyCom', \keyProof; \aux) \rightarrow 0/1; \aux'$.
            Verifies the commitment $\keyCom'$ against $\keyCom$ using the proof $\keyProof$ (and updates the internal state $\aux$ to $\aux'$).
    \end{itemize}
\end{defn}

Note that a signed update $\delta$ is bound to an epoch number $E$ to avoid replay attacks.
Further note that the provider's algorithms allow dividing up work: some tasks---\UpdateDB{} and \AddPK{}---are performed immediately upon receiving a client request, while others---finalizing and publishing $\dbCom$ and $\keyCom$---are deferred to the end of an epoch during \CloseEpoch{}.
This design may aid efficient implementations.
Finally, note that \EndEpoch{} can also generate initial \enquote{empty} commitments on startup.

We now turn to the desired security and privacy properties. The primary goal of a \ac{prim} scheme is to prevent a malicious provider from outputting a commitment to a database in which any user's entry is different from what the user has authorized.
If the consistency proofs $\dbProof$ and $\keyProof$ pass auditor verification, then every individual balance change must have been authorized by a user request.
Thus, in a \ac{prim}-based \ac{model}, an OTB attack leaves cryptographic traces.

To formalize this, we consider a security
experiment $\gameDbSec{}$ where the adversary controls the database entries of all but one honest user and has access to a signing oracle for the honest user.
The adversary must output a consistent history, of \emph{any} length, of commitments and proofs (that auditors would accept) that eventually leads to a corrupted database. While only the final corrupted commitment needs to be output, note that the adversary may attempt to craft intermediate commitments that modify the internal state of either the honest user or the auditor. Valid inclusion proofs must also be provided for the honest user following each signing query.
$\gameDbSec$ proceeds as follows:
\begin{enumerate}
    \item Public parameters \pp{} and a key pair $\sk_i, \pk_i$ are generated by $\Setup{}(\secpar, n)$ and, respectively, $\KeyGen{}(i)$.
        An empty list $\expList$ is initialized.
    \item \pp{} and $\pk_i$ are sent to the adversary.
    \item The adversary gets access to a \Sign{} oracle for $\sk_i$.
        On input $(j, \delta^{(j)})$, the oracle outputs $\sig^{(j)}$ as computed by $\Sign(\sk_i, j, \delta^{(j)})$.
        The query is then added to the list $\expList$.
        However, if previously there already was another query of the form $(j, \cdot)$, then the same output as before is returned and $\expList$ stays unchanged.
    \item Eventually, the adversary outputs integers $K, E$ with $E < K$, commitments $\keyCom^{(0)}, \dbCom^{(0)}$, commitments and auditor proofs $\dbCom^{(j)}, \keyCom^{(j)}, \dbProof^{(j)}, \keyProof^{(j)}$ for all $j \in [K]$, a key inclusion proof $\keyProof_i^{(E)}$, data inclusion proofs $\dbProof_i^{(j)}$ for all $j \in (\expList \cap [E+1,K-1]) \cup \{K\}$ and a value $z$.
    \item The adversary wins if:
        \begin{itemize}
            \item $\forall j \in [K]$, $\VerifyKeys(\keyCom^{(j-1)},\allowbreak{} j,\allowbreak{} \keyCom^{(j)},\allowbreak{} \keyProof^{(j)};\allowbreak{} \aux_1^{(j-1)})\allowbreak{} = 1$ where $\aux_1^{(j-1)}$ gets updated to $\aux_1^{(j)}$ and $\aux_1^{(0)}$ is the auditor's (honest) initial state.
            \item $\forall j \in [K]$, $\VerifyDB(\keyCom^{(j-1)},\allowbreak{} j,\allowbreak{} \dbCom^{(j)},\allowbreak{} \dbProof^{(j)};\allowbreak{} \aux_2^{(j-1)})\allowbreak{} = 1$ where $\aux_2^{(j-1)}$ gets updated to $\aux_2^{(j)}$ and $\aux_2^{(0)}$ is the auditor's (honest) initial state.
            \item $\VerifyPK(\keyCom^{(E)}, i, \pk_i, \keyProof_i) = 1$.
            \item $\forall j \in \expList \cap [E+1, K-1]$, $\VerifyLU(\dbCom^{(j)},\allowbreak{} i,\allowbreak{} \sum_{k \in \expList \cap [E+1, j]} \delta^{(k)},\allowbreak{} \dbProof_i^{(j)})\allowbreak{} = 1$.
            \item $\VerifyLU(\dbCom^{(K)}, i, z, \dbProof_i^{(K)}) = 1$.
            \item $z \neq \sum_{k \in \expList \cap [E+1, K]} \delta^{(k)}$.
        \end{itemize}
\end{enumerate}

\begin{defn}[Data Security]\label{def:dbSec}
    Let $\secpar$ be a security parameter.
    A \ac{prim} scheme has data security if the probability that any probabilistic polynomial-time adversary wins $\gameDbSec{}$ is bounded by a negligible function in $\secpar$.
\end{defn}

For privacy, it is required that neither the commitments nor the per-user inclusion proofs nor the global auditor proofs leak the content of $\dbCom$. More precisely, we define privacy as indistinguishability of a real protocol flow from a simulated one by an adversary $\advA$ that may control the database state, arbitrarily many corrupted users and auditors. In the real flow (formally $\gameZkReal{}$), the provider behaves normally. In the simulated flow (formally $\gameZkSim$), the provider's messages come from a simulator algorithm that is \emph{not} given honest user's private data.
The real protocol flow, $\gameZkReal$, proceeds as follows:
\begin{enumerate}
    \item Compute $\pp$ using \Setup{}.
        Send $\pp$ to $\advA$.
        Set $E=0$ and $I,J = \emptyset$.
    \item Arbitrarily many times, $\advA$ chooses between:
        \begin{itemize}
            \item \textbf{Adversarial user:}
                Send $i$. Receive $\pk_i$. Call $\AddPK(i, \pk_i)$.
                Set $I = I \cup \{i\}$.
            \item \textbf{Honest user:}
                Compute $(\sk_i, \pk_i) = \KeyGen(i)$.
                Call $\AddPK(i, \pk_i)$.
                Set $J = J \cup \{i\}$.
                Send $i, \pk_i$.
            \item \textbf{Adversarial update:}
                Receive $i, \delta, \sig$. If $i \in I$ and $\VerifySig(i, E, \pk_i, \delta, \sig)$ accepts, call $\UpdateDB(i, \delta, \sig)$.
            \item \textbf{Honest update:}
                Receive $i, \delta$. If $i \in J$, compute $\sig = \Sign(\sk_i, E, \delta)$ and call $\UpdateDB(i, \delta, \sig)$.
            \item \textbf{End epoch:}
                Call \CloseEpoch{} and send the output $\left(\dbCom,\allowbreak{} \keyCom,\allowbreak{} \dbProof,\allowbreak{} \keyProof,\allowbreak{} \left(\dbProof_i\right)_{i \in [0,n-1]},\right.$ $\left.\left(\keyProof_i\right)_{i \in [0,n-1]}\right)$. Set $E = E+1$.
        \end{itemize}
    \item $\advA$ outputs a bit.
\end{enumerate}
The simulated flow, relative to the simulator $\Sim = (\Sim_1, \Sim_2, \Sim_3, \Sim_4, \Sim_5, \Sim_6)$, is defined as the experiment $\gameZkSim$ (during which $\Sim$ may read and write auxiliary state):
\begin{enumerate}
    \item Compute $\pp$ using $\Sim_1$.
        Send $\pp$. Set $E=0$ and $I,J = \emptyset$.
    \item Arbitrarily many times, $\advA$ chooses between:
        \begin{itemize}
            \item \textbf{Adversarial user:}
                Send $i$. Receive $\pk_i$.
                Call $\Sim_2(i, \pk_i)$.
                Set $I = I \cup \{i\}$.
            \item \textbf{Honest user:}
                Call $\Sim_3(i)$, which outputs $\pk_i$.
                Set $J = J \cup \{i\}$.
                Send $i, \pk_i$.
            \item \textbf{Adversarial update:}
                Receive $i, \delta, \sig$. If $i \in I$ and $\VerifySig(i, E, \pk_i, \delta, \sig)$ accepts, call $\Sim_4(i, \delta, \sig)$.
            \item \textbf{Honest update:}
                Receive $i, \delta$. If $i \in J$, call $\Sim_5(i)$.
            \item \textbf{End epoch:}
                Call $\Sim_6$, which outputs $\left(\dbCom,\allowbreak{} \keyCom,\allowbreak{} \dbProof,\allowbreak{} \keyProof,\allowbreak{}\right.$ $\left(\dbProof_i\right)_{i \in [0,n-1]},\allowbreak{}$ $\left.\left(\keyProof_i\right)_{i \in [0,n-1]}\right)$, and send the output. Set $E = E+1$.
        \end{itemize}
    \item $\advA$ outputs a bit.
\end{enumerate}

\begin{defn}[Data Privacy]\label{def:dbPriv}
    Let $\secpar$ be a security parameter.
    A \ac{prim} scheme has data privacy if there is a simulator such that for any probabilistic polynomial-time adversary, $|\Pr[\gameZkReal = 1] - \Pr[\gameZkSim = 1]|$ is bounded by a negligible function in $\secpar$.
\end{defn}

Note that a PVC does not explicitly try to hide metadata, like the number of users $n$.
Indeed, hiding $n$ is unnecessary in the \ac{model} context as the customer count of crypto exchanges is made public in \emph{legally mandated reports} (e.g., \cite{coinbaseReport}).

\section{Preliminaries \& New Building Blocks}\label{sec:prelims}

\subsection{Conventions}
$\FF$ is a prime field.
$\GG, \hat{\GG}, \GG_T$ are elliptic curve groups of order $|\FF|$ with generators $g \in \GG, \hat{g} \in \hat{\GG}$ and with a bilinear map $e: \GG \times \hat{\GG} \rightarrow \GG_T$.
The group operation is written as $+$.
We use brackets for clarity, e.g., denoting by $\delta \cdot [\tau \cdot g]$ the scalar multiplication of $\delta \in \FF$ with the $\GG$-element whose logarithm is $\tau$.
A multiplicative subgroup of size $n$, $\subgroup = \{1, \omega, \dots, \omega^{n-1}\} \subset \FF$, is the interpolation domain for Lagrange polynomials.
The associated basis is $\lag_0(x), \dots, \lag_{n-1}(x)$.
Polynomials are denoted with lower-case letters, while commitments use upper-case letters.
We define $\perm$ as the natural FFT ordering on $\subgroup$, i.e., the permutation on $[0,n-1]$ which maps $i \mapsto \perm(i) = \sum_{j \in [0,\log(n)-1]} i_j \cdot 2^{\log(n)-j-1}$ where $i_0, \dots, i_{\log(n)-1} \in \{0,1\}$ are the bits of $i$. \emph{Lookup}, \emph{inclusion} and \emph{opening} proof are used as synonyms.

\subsection{KZG Vector Commitments}\label{sec:kzg}

A vector commitment \cite{cf13} is a short digest $A$ of an list of values $a_0, \dots, a_{n-1}$ with efficient proving and verification algorithms for statements of the form \enquote{the $i$th element in the list is $a_i$} (given any $i$).
A polynomial commitment \cite{kzg} is a short digest $\dbCom$ of a polynomial $\dbPoly$ with efficient proving and verification algorithms for statements of the form \enquote{$\dbCom$ commits to a polynomial $\dbPoly$ such that $\dbPoly(y) = z$} (given $y$ and $z$).
Throughout, we identify a vector ${a} = (a_0,\dots,a_{n-1})$ with the polynomial $\dbPoly$ with $\dbPoly(\omega^i) = a_i$ and we use a commitment to $\dbPoly$ as a commitment to ${a}$ \cite[cf.][]{asvc}.

We use the following (\enquote{AMT}) variant \cite{tomescu20} of the KZG scheme \cite{kzg}, comprising five algorithms.
\begin{itemize}
    \item $\KZGSetup(\secpar, n)$, given a security parameter $\secpar$ and an integer $n$, outputs global public parameters in the form of a structured reference string (SRS) consisting of $\tau \cdot g, \dots, \tau^{n-1} \cdot g, \tau \cdot \hat{g}, \dots, \tau^{n} \cdot \hat{g}$ (without revealing $\tau$).
    \item $\KZGCommit(\dbPoly)$ takes the values of a polynomial $\dbPoly$ over $\subgroup$ and outputs the commitment $\dbCom \in \GG$ as $\dbPoly(\tau) \cdot g = \sum_i \dbPoly(\omega^i) \cdot [\lag_i(\tau) \cdot g]$ or, alternatively, computes the commitment $\hat{\dbCom} \in \hat{\GG}$ as $\dbPoly(\tau) \cdot \hat{g}$ in analogous fashion.
    \item $\AMTOpen(\dbCom, i)$, given any $i \in [0,n-1]$, outputs the (precomputed) proof that $\dbPoly(\omega^i)=z$.
        The proof is $\dbProof_i = (\dbProof_{i,0}, \dots, \dbProof_{i,\log(n)-1}) \in \GG^{\log(n)}$.
    \item $\AMTVerify(\dbCom,i,Z,\dbProof)$ checks, with $Z = z \cdot g$:
        \[e(\dbCom - Z, \hat{g}) = \sum_{j \in [0,\log(n)-1]} e\left(\dbProof_{i,j}, [\tau^{2^j} \cdot \hat{g}] - \omega^{i \cdot 2^j} \cdot \hat{g}\right).\]
        For a $\hat{\GG}$-commitment, the left-hand side is $e(g, \hat{\dbCom} - z \cdot \hat{g})$.
\end{itemize}
The opening proofs form a binary tree since, e.g., half of all $\omega^i$ share the same last proof component $\dbProof_{i,\log(n)-1}$.
If a single value $\dbPoly(\omega^i)$ is modified, i.e., if $\dbPoly$ is updated to $\dbPoly' = \dbPoly + \delta \cdot \lag_i$, then the commitment can be updated as $\dbCom' = \dbCom + \delta \cdot [\lag_i(\tau) \cdot g]$.
Let $(\lagProof_{i,0}, \dots, \lagProof_{i,\log(n)-1})$ be the proof that $\lag_i(\omega^i)=1$.
The final algorithm which updates all opening proofs (and also creates the tree in the first place) is:
\begin{itemize}
    \item $\AMTMaintain(\dbCom, i, \delta)$ updates the binary tree of proof components in $\log(n)$ places as $\dbProof_{i,j}' = \dbProof_{i,j} + \delta \cdot \lagProof_{i,j}$.
\end{itemize}

\subsection{Quotient Proofs}\label{sec:quotients}

We use the following well-known algebraic identities to prove statements about interpolated polynomials.

\subheading{Zerocheck \cite[cf.][]{aurora}:}
If $\dbPoly, \keyPoly, \sigPoly$ are polynomials of degree $\leq n-1$, then $\dbPoly(\omega^i) \cdot \keyPoly(\omega^i) = \sigPoly(\omega^i)$ for all $i \in [0,n-1]$ if and only if there exists $\zeroQuo, \deg(\zeroQuo) \leq n-2$ such that $\dbPoly \cdot \keyPoly - \sigPoly = \zeroQuo \cdot (x^n-1)$.

\subheading{Binarity:}
If $\binPoly$ is a polynomial of degree $\leq n-1$, then $\binPoly(\omega^i) \in \{0,1\}$ for all $i \in [0,n-1]$ if and only if there exists $\binQuo,\deg(\binQuo)\leq n-2$ such that $\binPoly - \binPoly^2 = \binQuo \cdot (x^n-1)$.

\subheading{Sumcheck \cite[cf.][]{aurora}:}
If two polynomials $\keyPoly, \binPoly$ are of degree $\leq n-1$, given $\total \in \FF$, we have $\total = \sum_i \keyPoly(\omega^i) \cdot \binPoly(\omega^i)$ if and only if there exist polynomials $\apkRem,\apkQuo$ with $\deg(\apkRem),\deg(\apkQuo) \leq n-2$ such that $\keyPoly \cdot \binPoly = \total/n + \apkRem \cdot x + \apkQuo \cdot (x^n-1)$.

\subsection{APK Proofs}\label{sec:secrets}\label{sec:apkproof}

\subheading{Computing $\keyCom$.}
As per \cite{garg24,das23}, to commit to $n$ secret values---e.g., signing keys in the BLS signature scheme \cite{bls04}---with KZG, for each $i \in [0,n-1]$, the holder of the secret $\sk_i$ makes $\sk_i \cdot \lag_i(\tau) \cdot g$ public.
Then, the sum of these values is the KZG commitment to a polynomial $\keyPoly$ which interpolates the $n$ secrets: $\keyCom = \sum_{i \in [0,n-1]} \sk_i \cdot \lag_i(\tau) \cdot g.$

\subheading{APK proof.}
An \emph{APK proof} \cite{garg24,das23,cssv} proves statements of the form \enquote{$\exists I \subseteq [0,n-1]$ such that $\apk = \sum_{i \in I} \sk_i \cdot g$}, given a group element $\apk$.
The terminology stems from BLS signatures where the quantity \apk{} is the \emph{aggregated public key} of the given subset of signers.
Known KZG-based APK proofs \cite{garg24,das23}, roughly, consist of commitments to a binary-valued polynomial $\binPoly$, to a binarity proof quotient $\binQuo$ as in \cref{sec:quotients} and commitments to sumcheck quotients $\apkRem$ and $\apkQuo$ such that $\keyPoly \cdot \binPoly = \left(\sum_{i \in I} \sk_i\right)/n + \apkRem \cdot x + \apkQuo \cdot (x^n-1).$
This identity is verified using the pairing, which is possible since the sum in the above equation must be exactly the discrete logarithm of $\apk$.
For a prover to be able to construct this proof, each secret holder makes additional helpers values public.
Namely, each signer $i$ makes public
\[\zeroQuoCom_{i,i} = \sk_i \cdot \frac{\lag_i(\tau)^2-\lag_i(\tau)}{\tau^n-1} \cdot g,\ \apkRemCom_i = \sk_i \cdot \frac{\lag_i(\tau)-1/n}{\tau} \cdot g\]
and, for all $j \neq i$, \[\zeroQuoCom_{i,j} = \sk_i \cdot \frac{\lag_i(\tau) \cdot \lag_j(\tau)}{\tau^n-1} \cdot g.\]
With $\zeroQuoCom_i = \sum_{j \in [0,n-1]} \zeroQuoCom_{j,i}$, commitments to $\apkQuo$ and $\apkRem$ can be computed as $\sum_{i \in I} \zeroQuoCom_{i}$ and, respectively, $\sum_{i \in I} \apkRemCom_i.$
An explicit degree check must be performed on $\apkRem$, which can use the polynomial $\apkRemDeg(x) = \apkRem(x) \cdot x$ with commitment $\apkRemDegCom = \sum_{i \in I} \left([\sk_i \cdot \lag_i(\tau) \cdot g] - 1/n \cdot [\sk_i \cdot g]\right)$ \cite{garg24}.
The verification, which we denote here as $\APKVerify(\keyCom, \apk, (\binCom, \hat{\binCom}, \binQuoCom, \apkQuoCom, \apkRemCom, \apkRemDegCom))$, checks
\textbf{(1)} $e(\binCom, \hat{g}) = e(g, \hat{\binCom})$,
\textbf{(2)} $e(\binCom, \hat{g} - \hat{\binCom}) = e(\binQuoCom, [\tau^n \cdot \hat{g}] - \hat{g})$,
\textbf{(3)} $e(\apkRemCom, [\tau \cdot \hat{g}]) = e(\apkRemDegCom, \hat{g})$ and
\textbf{(4)} $e(\keyCom, \hat{\binCom}) = e(\apk_E/n, \hat{g}) + e(\apkRemCom, [\tau \cdot \hat{g}]) + e(\apkQuoCom, [\tau^n \cdot \hat{g}] - \hat{g})$.

\subheading{Helper values' well-formedness.}
The helper values are verified against the users' public keys using a pairing check, e.g., $e(\apkRemCom_i, \hat{g}) = e\left(\pk_i, \left[\frac{\lag_i(\tau)-1/n}{\tau} \cdot \hat{g}\right]\right)$ for $\apkRemCom_i$.
Formally, we will abbreviate all such checks on helper values with a single \VerifyHelpers{} routine.

\subsection{Lookup Tree}\label{sec:secretTree}

We are interested in efficient updatable lookup proofs for $\keyCom$, enabling statements like “The $i$th key in $\keyCom$ is $\sk_i$” for all $i \in [0, n-1]$, with only $O(\log n)$ update cost when a single entry changes.
These are precisely the properties of the proof tree from \cref{sec:kzg}, prompting the question: can it be adapted to the setting of committed secrets?
We observe that such a construction is, in fact, possible given only $\log(n)$ additional helper values from each signer.
Specifically, each secret holder $i$ may make public, for $j \in [0,\log(n)-1]$, $\sk_i \cdot L_{i,j}$ where $L_{i,j}$ are the proof components for $\lag_i(\omega^i)=1$ from \cref{sec:kzg}.
Then, by linear interpolation, $\keyProof_{i,j} = \sum_{i \in [0,n-1]} [\sk_i \cdot L_{i,j}]$ must be the components of a sound opening proof for $\keyPoly(\omega^i) = \sk_i$ in the sense of \AMTVerify{}. In our context, the only verifier of $\sk_i$ is the holder of $\sk_i$.
Therefore, the \AMTVerify{} algorithm can be used exactly as is.
As before, all $n$ opening proofs can be organized as a binary tree. An update to the secret commitment of the form $\keyCom' = \keyCom + [\sk_i \cdot \lag_i(\tau) \cdot g]$ requires updating the tree in only $\log(n)$ places as $\keyProof_{i,j}' = \keyProof_{i,j} + [\sk_i \cdot L_{i,j}]$.
This way, the proof tree's properties indeed carry over to the setting of secret polynomials, as desired.

We overload \AMTMaintain{} for commitments to secrets as described here.

\subsection{Append-Only Updates}\label{sec:append}

\newlength\levelFourDist
\setlength\levelFourDist{0.4cm}
\newlength\levelThreeDist
\setlength\levelThreeDist{2\levelFourDist}
\newlength\levelTwoDist
\setlength\levelTwoDist{4\levelFourDist}
\newlength\levelOneDist
\setlength\levelOneDist{8\levelFourDist}
\newlength\tableColMinWitdth
\setlength\tableColMinWitdth{1.95\levelFourDist}
\newlength\bracesOffset
\setlength\bracesOffset{0.427cm}

\begin{figure}[!t]
    \centering
    \begin{tikzpicture}[
            level distance=1cm,
            level 1/.style={level distance=0.5cm,sibling distance=\levelOneDist},
            level 2/.style={level distance=1cm, sibling distance=\levelTwoDist},
            level 3/.style={level distance=1cm, sibling distance=\levelThreeDist},
            level 4/.style={level distance=1cm, sibling distance=\levelFourDist},
            decoration={brace,amplitude=0.15cm}
        ]
        \node (root) {$\keyPoly$}
        child {node {$\keyPoly_{k,2}$}
            child {node {$\keyPoly_{k,1}$}
                child {node (sZeroZero) {$\keyPoly_{0,0}$}}
                child {node[rectangle,draw] {$\keyPoly_{k,0}$}}
            }
            child {node {$\keyPoly_{k,1}^*$ ($=0$)}
                child {node {.}}
                child {node {.}}
            }
        }
        child {node {$\keyPoly_{k,2}^*$ ($=0$)}
            child {node {.}
                child {node {.}}
                child {node {.}}
            }
            child {node {.}
                child {node {.}}
                child {node {.}}
            }
        };
        \matrix[matrix of nodes,
            anchor=north west,
            xshift=-2.1\tableColMinWitdth, nodes={minimum width=\tableColMinWitdth, anchor=center},
            nodes in empty cells
        ] (table) at (sZeroZero.south){$i$: & 0 & $\frac{n}{2}$ & $\cdots$ &  &  &  & $\frac{n}{2}-1$ & $n-1$ \\
            $\perm(i)$: & 0 & $ 1 $ & $2$ & $3$ & $4$ & $5$ & $6$ & $7$ \\
            $\keyPoly(\omega^i)$: & $\neq0$ & $=0$ & $=0$ & $=0$ & $=0$ & $=0$ & $=0$ & $=0$ \\
            \textbf{Proof:} & &  &  &  &  &  &  & \\
        };
        \draw[decorate,transform canvas={yshift=\bracesOffset}] (table-4-3.south east) -- node[below=2pt] {$\keyProof_{k,0}$} (table-4-3.south west);
        \draw[decorate,transform canvas={yshift=\bracesOffset}] (table-4-5.south east) -- node[below=2pt] {$\keyProof_{k,1}$} (table-4-4.south west);
        \draw[decorate,transform canvas={yshift=\bracesOffset}] (table-4-9.south east) -- node[below=2pt] {$\keyProof_{k,2}$} (table-4-6.south west);
    \end{tikzpicture}
    \vspace{-2em}
    \caption{Example of our method to prove that $\keyPoly(\omega^i)=0$ for all $i$ with $\perm(i) \in [\perm(k),n-1]$ for $n=8,k=n/2$.
    \label{fig:tree}}
\end{figure}

Given two committed secret polynomials $\keyPoly, \keyPoly_E$, we are interested in proving that $\keyPoly+\keyPoly_E$ is an append-only update in the sense that $\keyPoly_E$ only adds new evaluations on $\subgroup$ to $\keyPoly$ \enquote{at the end}.
Formally, we will prove that that there exists a \emph{split point} $k \in [0,n-1]$ such that, firstly, $\keyPoly(\omega^i) = 0$ for all $i \in [0,n-1]$ with $\perm(i) \in [\perm(k), n-1]$ and, secondly, $\keyPoly_E(\omega^i) = 0$ for all $i$ with $\perm(i) \in [0,\perm(k)-1]$.
Here, $\perm$ is the bit reversal map, the standard FFT ordering for elements in $\subgroup$ \cite{clrs}.
See also \cref{fig:tree} for an example diagram.

The key idea is to leverage the proof tree from \cref{sec:secretTree}, once for $\keyPoly$ and once for $\keyPoly_E$.
In particular, we can take the proof components for $\keyPoly(\omega^k)=0$---the proof tree path consisting of $\keyProof_{k,0},\allowbreak{} \dots,\allowbreak{} \keyProof_{k,\log(n)-1}$---and exploit that each of these components has a \emph{dual nature}.
Let us define $\keyPoly_{k,\log(n)} = \keyPoly$ and $\keyPoly_{k,0} = \keyPoly(\omega^k)$. Then, on the one hand, for all $j \in [0,\log(n)-1]$, $\keyProof_{k,j}$ is a commitment to a quotient polynomial that proves that some degree-$(2^{j+1}-1)$ polynomial $\keyPoly_{k,j+1}$ agrees with some degree-$(2^{j}-1)$ polynomial $\keyPoly_{k,j}$ on the vanishing set of $x^{2^j}-\omega^{k \cdot 2^j}$.
But on the other hand, by the construction of the proof tree, $\keyProof_{k,j}$ is \emph{also} a proof that $\keyPoly_{k,j+1}$ agrees with some polynomial $\keyPoly_{k,j}^*$ on the vanishing set of $x^{2^j}+\omega^{k\cdot 2^j}$.
Since all $\keyPoly_{i,j}$, for all $i$, make up a binary tree, we may call $\keyPoly_{k,j}^*$ the sibling of $\keyPoly_{k,j}$ (\cref{fig:tree} shows the tree structure).
Our desired appendage proof is equivalent to proving the following three statements, given $k$.
\textbf{(1)} All the $\keyPoly_{k,j}^*$ which are \emph{right siblings} are 0. \textbf{(2)} In the analogous remainder tree for $\keyPoly_E$, all \emph{left siblings} of polynomials on the path from $\keyPoly_E(\omega^{\perm(\perm(k)-1))})$ to $\keyPoly_E$ are 0.
\textbf{(3)} $\keyPoly(\omega^k)= \keyPoly_E(\omega^{\perm(\perm(k)-1))})=0$.
It follows from the above discussion that the opening proofs $(\keyProof_{k,0}, \dots, \keyProof_{k,\log(n)-1})$ for $\keyPoly(\omega^k)$ and, say, $(\newKeyProof_0, \dots, \newKeyProof_{\log(n)-1})$ for $\keyPoly_E(\omega^{\perm(\perm(k)-1))})$ can be used to confirm exactly these three claims due to the proofs' primary and dual nature.
A diagram of our method, for the case of check \textbf{(1)}, is shown in \cref{fig:tree}.
The full verification routine, \VerifyAppend{}, is shown in \cref{fig:append}.
Here, we use the fact that the bit decompositions of $k$ and $\perm(\perm(k)-1)$ determine whether polynomials are left or right children.
Our approach is simple and prover-efficient since, given the proof trees for $\keyPoly$ and $\keyPoly_E$, producing an appendage proof takes no extra work.\footnote{Other append-only proofs are possible, though likely less efficient. For example, one might use a quotient that proves the vanishing of $\keyPoly \cdot \keyPoly_E$ over $\domain$. }
When we write out the full verification equations (\cref{fig:append}), we see that $O(j)$ pairings are computed per tree level $j$, giving $O(\log^2(n))$ in total.

\begin{figure}[!t]
\centering
\begin{pcvstack}[boxed]
\procedure[bodylinesep=\belowLine]{$\VerifyAppend\left(\keyCom,\keyCom_E,k,\left(\keyProof_{k,j},\newKeyProof_j\right)_{j\in[0,\log(n)-1]}\right)$}{
c = \perm(\perm(k)-1) \\
\pcfor \ell \in [0,\log(n)-1]: \\
\quad \pcif k_{\log(n)-1-\ell} = 0: \pcComment{$\keyPoly_{k,\ell}$ is a left child} \\
\qquad \pccheck e(\keyCom,\hat{g}) = e\left(\keyProof_{k,\ell}, [\tau^{2^\ell} \cdot \hat{g}] + \omega^{k\cdot 2^{\ell}} \cdot \hat{g}\right) \\
\qquad\qquad\qquad + \textstyle\sum_{j \in [\ell+1,\log(n)-1]} e\left(\keyProof_{k,j}, [\tau^{2^j} \cdot \hat{g}] - \omega^{k\cdot 2^{j}} \cdot \hat{g}\right) \\
\quad \pcif c_{\log(n)-1-\ell} = 1: \pcComment{$\left(\keyPoly_E\right)_{c,\ell}$ is a right child} \\
\qquad \pccheck e(\keyCom_E,\hat{g}) = e\left(\newKeyProof_{\ell}, [\tau^{2^\ell} \cdot \hat{g}] + \omega^{c\cdot 2^{\ell}} \cdot \hat{g}\right) \\
\qquad\qquad\qquad + \textstyle\sum_{j \in [\ell+1,\log(n)-1]} e\left(\newKeyProof_{j}, [\tau^{2^j} \cdot \hat{g}] - \omega^{c\cdot 2^{j}} \cdot \hat{g}\right) \\
\pccheck \AMTVerify\left(\keyPoly, k, 0, \left(\keyProof_{k,j}\right)_{j\in[0,\log(n)-1]}\right) \\
\pccheck \AMTVerify\left(\keyPoly_E, c, 0, \left(\newKeyProof_{j}\right)_{j\in[0,\log(n)-1]}\right)
}
\end{pcvstack}
\caption{Verifying that $\keyPoly+\keyPoly_E$ is an append-only update.
Here, $k_0,\allowbreak{} \dots,\allowbreak{} k_{\log(n)-1}$ and $c_0,\dots,c_{\log(n)-1}$ are the bits of $k$ and $c$.
\label{fig:append}}
\end{figure}

\section{\NoCaseChange{\acs{prim}} Instantiation}\label{sec:sadConstruct}\label{sec:primConstruct}

\subsection{Construction Overview}\label{sec:overview}

We now introduce our \ac{prim} construction assuming a fixed table of keys (with public commitment $\keyCom$ defined next) and walk through the system's operations in an epoch $E$. Adding new users and further details are treated in \cref{sec:fullSAD} below.
Generally, all database entries are mapped to the finite field $\FF$ and updating an entry $v_i$ to $v_i'$ is seen as the addition of $\delta_i = v_i'-v_i$ to $v_i$.
All vector commitments are built using polynomial commitments and interpolation over $\subgroup = \{1, \omega, \dots, \omega^{n-1}\} \subset \FF$ with the Lagrange basis $\lag_0, \dots, \lag_{n-1}$.
The provider keeps KZG commitments $\dbCom \in \hat{\GG}, \keyCom, \sigCom \in \GG$ to three polynomials:
\begin{itemize}
\item For clients' values: $\dbPoly(x) = \sum_i v_i \cdot \lag_i(x)$, with an AMT proof tree for lookups.
\item For clients' keys: $\keyPoly(x) = \sum_i \sk_i \cdot \lag_i(x)$, with a proof tree as in \cref{sec:secretTree}.
\item Linking values to keys: $\sigPoly(x) = \sum_i \sk_i \cdot v_i \cdot \lag_i(x)$.
\end{itemize}

The clients' public keys are $\sk_i \cdot g \in \GG$.
As the secrets $\sk_i$ are unknown to the provider, $\keyPoly$ and $\sigPoly$ are only processed homomorphically using client-provided helper values (like $\sk_i \cdot \lag_i(\tau) \cdot g$).
Besides the polynomial commitments and the proof trees, the provider further maintains a {zerocheck} proof that $\dbPoly \cdot \keyPoly - \sigPoly$ vanishes over $\subgroup$.
This proof establishes that $\dbCom$ is consistent with $\sigCom$ and $\keyCom$.
$\sigCom$, in turn, is independently tied to $\keyCom$ through a specialized signature scheme.
In brief, when assuming that $\keyCom$ is correct, it is first enforced that $\sigCom$ only gets updated with values that were signed by users in a given epoch.
Then, via the zerocheck, $\dbCom$ is authenticated against $\sigCom$ by \enquote{factoring out} $\sk_i$ from each entry $\sk_i \cdot v_i$ in $\sigCom$.

\subheading{User and provider actions.}
Client $i$ requests an update in epoch $E$ by sending a signature $\sig_i$ over a desired update $\delta_i$ to his value, computed by \Sign{} as
\[
\sig_i = \sk_i \cdot (\groupHash(E) + \delta_i \cdot [\lag_i(\tau) \cdot \hat{g}]) \in \hat{\GG}
\]
where $\groupHash$ is a hash function to $\hat{\GG}$.
This signing algorithm is similar in spirit to the linearly homomorphic scheme from \cite{ap19}.
However, our scheme works over KZG commitments instead of field elements.\footnote{
The scheme from \cite{ap19} signs $\delta_i$ as $\sig_i = \sk_i \cdot (\groupHash(E) + \delta_i \cdot g^*)$ with a fixed base $g^*$.
We replace $g^*$ with $\lag_i(\tau) \cdot g$ and thereby sign a \emph{KZG commitment} to $\delta_i \cdot \lag_i(x)$.
With varying base points, the security guarantee from \cite{ap19} does not fully carry over and we cannot use $\sum_i \sig_i$ alone to verify $\dbCom$ (or $\sigCom$).
Hence, the zerocheck is needed. }
In particular, user $i$ signs an update $\delta_i$ using the $i$th Lagrange basis polynomial $\lag_i$.
It is assumed that each client makes at most one such balance update in an epoch.
If users wish to make frequent updates, the epoch length may be reduced.
Looking ahead, our evaluation (\cref{sec:eval}) indicates that the prover is indeed performant enough to keep up with short epochs.
Otherwise, multiple \enquote{real} updates could be aggregated into a single protocol-level update, which is the one the user would sign.
The user does \emph{not} sign multiple updates in one epoch, as the provider could otherwise tamper with the user's balance.

Let $I$ denote the set of indices of clients who updated their balances in this epoch.
From these updates, over the course of epoch $E$, the provider constructs:

\begin{itemize}
\item a commitment $\sigCom_E \in \GG$, for $\sigPoly_E(x) = \sum_{i \in I} \sk_i \cdot \delta_i \cdot \lag_i(x)$,
\item the aggregated public key $\apk_E = \sum_{i \in I} \sk_i \cdot g$,
\item an APK proof for $\apk_E$ with respect to $\keyPoly$ and
\item an aggregate signature
\[
\sig_E = \sum_{i \in I} \sig_i = \left( \sum_{i \in I} \sk_i \right) \cdot \groupHash(E) + \sigPoly_E(\tau) \cdot \hat{g}.
\]
\end{itemize}

When a client sends a request, the provider first checks the signature with \VerifySig{} (and that the value update is allowed according to an application-specific policy).
The provider then updates $\sigCom$, $\sigCom_E$, $\dbCom$, $\sig_E$, $\apk_E$ and the lookup, APK and zerocheck proofs (\UpdateDB{}).

At the end of the epoch, the provider finalizes the APK proof (as part of \Publish{}, lines~\ref{line:apkFinalizeStart}--\ref{line:apkFinalizeEnd} in \cref{fig:prim}).
This is the only inherent computational work for the provider at the end of an epoch---the other steps of proof creation can be performed online, as soon as client updates are issued.
The provider publishes $\sig_E$, $\apk_E$, the APK proof, $\sigCom_E$, $\dbCom$ and the zerocheck proof.

\subheading{Verification by Auditors (\VerifyDB{})} We assume that auditors remember the latest values $\sigCom$ and $\keyCom$.
First, they check the APK proof for $\apk_E$ with respect to $\keyCom$ and then use $\apk_E$ and $\sig_E$ to verify $\sigCom_E$ (line~\ref{line:sigVer} in \cref{fig:prim}).
This ensures that $\sigCom_E$ is a commitment to deltas signed by legitimate users.
Then, auditors update $\sigCom$ by adding $\sigCom_E$.
Finally, they verify the zerocheck for $\dbPoly \cdot \keyPoly - \sigPoly$, which makes sure that $\dbCom$ must contain the legitimate users' values (line~\ref{line:zerocheckVer} in \cref{fig:prim}).
The users who updated their balance query for inclusion proofs of their values in $\dbCom$ (\VerifyLU{}).
This is unavoidable to detect the omission of any particular update.

\subsection{Full Specification}\label{sec:fullSAD}\label{sec:fullPrim}

\begin{figure*}[!t]
\centering
\begin{pchstack}[space=0.3\baselineskip]
\begin{pcvstack}[boxed]
\begin{pchstack}[space=0.4\baselineskip]
    \begin{pcvstack}[space=0.5\baselineskip]
        \procedure[linenumbering,lnstart=0,bodylinesep=\belowLine]{$\Setup(\lambda, n)$}{
            \text{pick $\GG, \hat{\GG}, \GG_T$ with $\lambda$-hard co-CDH/$n$-SDH} \\
            \pp = \left(\left(\tau^i \cdot g, \tau^i \cdot \hat{g}, \tau^i \cdot h, \tau^i \cdot \hat{h}\right)_{i \in [0,n-1]}, \right.\pcskipln\\
            \quad\left. \tau^n \cdot \hat{g}, \tau^n \cdot h, \tau^n \cdot \hat{h}\right) \\
            \pcreturn \GG, \hat{\GG}, \GG_T, \pp
        }
        \procedure[linenumbering,lnstart=3,bodylinesep=\belowLine]{$\AddPK(i, \pk_i^*)$}{
            \pccheck i = \perm(\perm(k) + |J|) \\
            \pcparse \pk_i^* \text{ as } \left(\pk_i, \sk_i \cdot \lag_i(\tau) \cdot g, \right. \pcskipln\\
                \quad \left(\zeroQuoCom_{i,j}\right)_{j \in [0,n-1]}, \apkRemCom_i, \left(\sk_i \cdot \lagProof_{i,j}\right)_{j \in [0,\log(n)-1]}, \pcskipln\\
            \quad\ \left. \sk_i \cdot h, \sk_i \lag_i(\tau) \cdot h, \left(\widetilde{\zeroQuoCom_{i,j}}\right)_{j \in [0,n-1]}\right) \\
            \pccheck \VerifyHelpers(\pk_i^*) \pcskipln\\
            \pcComment{update hints} \\
            \apkRemDegCom_i = [\sk_i \cdot \lag_i(\tau) \cdot g] - [\sk_i \cdot g]/n \\
            \pcfor j \in [0,n-1]: \\
            \quad \zeroQuoCom_j = \zeroQuoCom_j + \zeroQuoCom_{i,j} \\
            \quad \widetilde{\zeroQuoCom_j} = \widetilde{\zeroQuoCom_j} + \widetilde{\zeroQuoCom_{i,j}} \pcskipln\\
            \pcCommentLine{update commitment and append proof} \\
            \keyCom = \keyCom + [\sk_i \cdot \lag_i(\tau) \cdot g] \\
            \AMTMaintain(\keyCom, i, \pk_i) \\ \keyCom_E = \keyCom_E + [\sk_i \cdot \lag_i(\tau) \cdot g] \\
            \AMTMaintain(\keyCom_E, i, \pk_i) \\
            J = J \cup \{i\}
        }
    \end{pcvstack}
    \begin{pcvstack}[space=0.5\baselineskip]
        \procedure[skipfirstln,linenumbering,lnstart=15,bodylinesep=\belowLine]{$\UpdateDB(i, \delta_i, (\sig_i, \epsilon_i))$}{
            \pcCommentLine{update datebase commitment}\\
            \label{line:updateDBStart} \dbCom^* = \dbCom^* + \delta_i \cdot [\lag_i(\tau) \cdot \hat{g}] \pcskipln\\
            \qquad\quad + \epsilon_i \cdot [\lag_i(\tau) \cdot \hat{h}] \\
            \AMTMaintain(\dbCom^*, i, (\delta_i, \epsilon_i)) \pcskipln\\
            \pcCommentLine{update signature and proof}\pcskipln\\
            \sigCom = \sigCom + \delta_i \cdot [\sk_i \cdot \lag_i(\tau) \cdot g] \pcskipln\\
            \qquad + \epsilon_i \cdot [\sk_i \cdot \lag_i(\tau) \cdot h] \\
            \zeroQuoCom = \zeroQuoCom + \delta_i \cdot \zeroQuoCom_i + \epsilon_i \cdot \widetilde{\zeroQuoCom_i} \\
            \sigCom_E = \sigCom_E + \delta_i \cdot [\sk_i \cdot \lag_i(\tau) \cdot g] \pcskipln\\
            \qquad\quad + \epsilon_i \cdot [\sk_i \cdot \lag_i(\tau) \cdot h] \\
            \sig_E = \sig_E + \sig_i \pcskipln\\
            \pcCommentLine{update APK and APK proof}\\
            \apk_E = \apk_E + \pk_i \\
            \binPoly = \binPoly + \lag_i \\
            \binCom = \binCom + [\lag_i(\tau) \cdot g] \\
            \hat{\binCom} = \hat{\binCom} + [\lag_i(\tau) \cdot \hat{g}] \\
            \apkRemCom = \apkRemCom + \apkRemCom_i \\
            \apkRemDegCom = \apkRemDegCom + \apkRemDegCom_i \\
            \apkQuoCom = \apkQuoCom + \zeroQuoCom_i \\
            \label{line:updateDBEnd} I = I \cup \{i\}
        }
        \procedure[linenumbering,lnstart=29,bodylinesep=\belowLine]{$\VerifySig(i, E, \pk_i, \delta_i, (\sig_i, \epsilon_i))$}{
            \pccheck e(g, \sig_i) = e(\pk_i, \hash(E) \pcskipln\\
                \quad + \delta_i \cdot [\lag_i(\tau) \cdot \hat{g}] \pcskipln\\
            \quad + \epsilon_i \cdot [\lag_i(\tau) \cdot \hat{h}])
        }
    \end{pcvstack}
\end{pchstack}
\begin{pchstack}[space=-0.085cm]
    \procedure[linenumbering,lnstart=30,bodylinesep=\belowLine]{$\Publish$}{
        \label{line:apkFinalizeStart} \text{find } \binQuo \text{ with } \binPoly-\binPoly^2 = \binQuo \cdot (x^n-1) \phantom{ab}\\
        \binQuoCom = \KZGCommit(\binQuo) \\
        \label{line:apkFinalizeEnd} \apkProof = (\binCom, \hat{\binCom}, \binQuoCom, \apkQuoCom, \label{line:afterApkFinalize} \apkRemCom, \apkRemDegCom) \\
        \dbProof = (\sigCom_E, \sig_E, \apk_E, \apkProof, \zeroQuoCom) \\
        \newKeyProof = \AMTGetProof(\keyCom_E, c) \\
        \keyProof = (\keyCom_E, k, \keyProof_k, \newKeyProof) \\
        \pcfor i \in I \cup J: \\
        \quad \dbProof_i = \AMTGetProof(\dbCom^*, i) \\
        \pcfor i \in J: \\
    \quad \keyProof_i = \AMTGetProof(\keyCom, i) }
    \procedure[linenumbering,skipfirstln,lnstart=40,bodylinesep=\belowLine]{\phantom{abc}}{
        \pcCommentLine{prepare next proofs} \\
        k = \perm(\perm(k) + |J|) \\
        \keyProof_k = \AMTGetProof(\keyCom, k) \\
        c = \perm(\perm(k)-1) \pcskipln\\
        \pcCommentLine{reset variables} \\
        \label{line:providerVarsStart} I, J = \emptyset \\
        \binPoly = 0 \in \FF[x] \\
        \sigCom_E, \binCom, \binQuoCom, \apkRemCom, \apkRemDegCom, \apkQuoCom, \apk_E, \keyCom_E = 0 \in \GG \\
        \label{line:providerVarsEnd} \sig_E, \hat{\binCom} = 0 \in \hat{\GG} \\
        \pcreturn \dbCom^*, \keyCom, \dbProof, \keyProof, \left(\dbProof_i\right)_{i \in I}, \label{line:publishEnd} \left(\keyProof_i, \dbProof_i\right)_{i \in J}
    }
\end{pchstack}
\end{pcvstack}

\begin{pcvstack}[space=0.43\baselineskip]
\begin{pcvstack}[boxed,space=0.55\baselineskip]
    \procedure[linenumbering,lnstart=48,bodylinesep=\belowLine]{$\KeyGen(i)$}{
        \sk_i \sample \FF \\
        \pk_i = \sk_i \cdot g \\
        \pk_i^* = \left(\pk_i, \sk_i \cdot [\lag_i(\tau) \cdot g], \left(\zeroQuoCom_{i,j}\right)_{j \in [0,n-1]}, \right.\pcskipln\\
            \quad\left. \apkRemCom_i, \left(\sk_i \cdot \lagProof_{i,j}\right)_{j \in [0, \log(n)-1]}, \sk_i \cdot h, \right.\pcskipln\\
        \quad\left. \sk_i \cdot [\lag_i(\tau) \cdot h], \left(\widetilde{\zeroQuoCom_{i,j}}\right)_{j \in [0,n-1]}\right) \\
        \pcreturn \sk_i, \pk_i^*
    }
    \procedure[linenumbering,lnstart=52,bodylinesep=\belowLine]{$\VerifyPK(\keyCom, i, \pk_i, (\keyProof_i, \dbProof_i); \dbCom^*)$}{
        \pccheck \AMTVerify(\keyCom, i, \pk_i, \keyProof_i) \\
        \pccheck \VerifyLU(\dbCom^*, i, (0,0), \dbProof_i)
    }
    \procedure[linenumbering,lnstart=54,bodylinesep=\belowLine]{$\Sign(\sk_i, E, \delta_i; \maskPoly_i)$}{
        \text{pick random } \epsilon_i \\
        \maskPoly_i = \maskPoly_i + \epsilon_i \\
        \sig_i = \sk_i \cdot \left(\hash(E) + \delta_i \cdot [\lag_i(\tau) \cdot \hat{g}] \right.\pcskipln\\
        \qquad\qquad\left. + \epsilon_i \cdot [\lag_i(\tau) \cdot \hat{h}]\right) \\
        \pcreturn \sig_i, \epsilon_i
    }
    \procedure[linenumbering,lnstart=58,bodylinesep=\belowLine]{$\VerifyLU(\dbCom^*, i, \dbPoly_i, \dbProof_i; \maskPoly_i)$}{
        \pccheck \AMTVerify(\dbCom^*, i, v_i \cdot \hat{g} + \maskPoly_i \cdot \hat{h}, \dbProof_i)
    }
\end{pcvstack}
\begin{pcvstack}[boxed,space=0.5\baselineskip]
    \procedure[linenumbering,lnstart=59,bodylinesep=\belowLine]{$\VerifyDB(\keyCom, E, \dbCom^*, \dbProof; \sigCom)$}{
        \label{line:verifyDbStart} \pcparse \dbProof \text{ as } (\sigCom_E, \sig_E, \apk_E, \apkProof, \zeroQuoCom) \\
        \pccheck \APKVerify(\keyCom, \apk_E, \apkProof) \\
        \label{line:sigVer} \pccheck e(g, \sig_E) = e(\apk_E, \hash(E)) + e(\sigCom_E, \hat{g}) \\
        \label{line:zerocheckVer} \pccheck e(\keyCom, \dbCom^*) = e(\sigCom + \sigCom_E, \hat{g}) \pcskipln\\
        \quad + e(\zeroQuoCom, [\tau^n \cdot \hat{g}] - \hat{g}) \\
        \label{line:verifyDbEnd} \sigCom = \sigCom + \sigCom_E
    }
    \procedure[linenumbering,lnstart=64,bodylinesep=\belowLine]{$\VerifyKeys(\keyCom, E, \keyCom', \keyProof)$}{
        \pcparse \keyProof \text{ as } \left(\keyCom_E, k, \left(\keyProof_{k,j}\right)_{j \in [0,\log(n)-1]}, \right.\pcskipln\\
        \quad\left. \left(\newKeyProof_j\right)_{j \in [0,\log(n)-1]}\right) \\
        \pccheck \VerifyAppend\left(\keyCom, \keyCom_E, k, \right.\pcskipln\\
        \quad\left. \left(\keyProof_{k,j}, \newKeyProof_{j}\right)_{j \in [0,\log(n)-1]}\right) \\
        \pccheck \keyCom' = \keyCom + \keyCom_E
    }
\end{pcvstack}
\end{pcvstack}
\end{pchstack}
\vspace{-1em}
\caption{
\Ac{prim} algorithms for the provider, for users and for auditors.
\label{fig:prim}
\label{fig:prim2}
}
\end{figure*}

To describe our scheme fully (\cref{fig:prim}), we specify the zerocheck proof as a commitment $\zeroQuoCom \in\GG$ to a polynomial $\zeroQuo$ as in \cref{sec:quotients}. The APK proof $\apkProof$ consists of commitments $\binCom, \binQuoCom, \apkRemCom, \apkQuoCom, \apkRemDegCom \in \GG, \hat{\binCom}\in\hat{\GG}$ for, respectively, a binary-valued polynomial $\binPoly$, to a binarity proof $\binQuo$, to sumcheck quotients $\apkRem$ and $\apkQuo$ and to a degree checker $\apkRemDeg$---all defined as in \cref{sec:apkproof}.
Initially, these are 0.

For the provider to construct and update $\zeroQuoCom, \apkRemCom, \apkRemDegCom$ and $\apkQuoCom$, users provide helper values, from which $\zeroQuoCom_i, \apkRemCom_i$ are computed as in \cref{sec:apkproof}.
Defining $\zeroQuo_i$ as the polynomial committed in $\zeroQuoCom_i$, using interpolation of $\keyPoly$ and rearranging:
\begin{align*}
\lag_i \cdot \keyPoly
&= \sum_{j \in [0,n-1]} \lag_i \cdot \lag_j \cdot \sk_j  \\
&= \sk_i \cdot \lag_i + (\lag_i^2-\lag_i) \cdot \sk_i + \sum_{j \in [0,n-1] \setminus \{i\}} \lag_i \cdot \lag_j \cdot \sk_j \\
&= \sk_i \cdot \lag_i + \zeroQuo_i \cdot (x^n-1).
\end{align*}
Hence, when an update $\delta_i$ is applied to $v(\omega^i)$, we can write
\[
(\dbPoly + \delta_i \cdot \lag_i) \cdot \keyPoly = (\sigPoly + \delta_i \cdot \sk_i \cdot \lag_i) + (\zeroQuo + \delta_i \cdot \zeroQuo_i) \cdot (x^n-1).
\]
This yields a joint update formula for $\zeroQuoCom$ and $\sigCom$, based on $\zeroQuoCom_i$, when $\delta_i$ is added to $\dbPoly$.
Similarly, $\apkQuoCom$ can also be updated by adding $\zeroQuoCom_i$.
To update $\apkRemCom$, it suffices to add $\apkRemCom_i$.
To update $\apkRemDegCom$, it suffices to add $\left([\sk_i \cdot \lag_i(\tau) \cdot g] - 1/n \cdot \pk_i\right)$.
Now, $\binQuoCom$ cannot be computed incrementally due to the quadratic relationship of $\binQuo$ with $\binPoly$.
Instead, $\binQuoCom$ gets computed from scratch at the end of an epoch using an $O(n \cdot \log(n))$ Fast Fourier Transform approach (representing the only inherent end-of-epoch computation overhead).

\subheading{Adding users.}
To register, a new user sends a request and is given a free index $i$.
Then, the user locally generates a key with \KeyGen{} and sends the formal public key $\pk_i^*$ to the server, where $\pk_i^*$ includes the proper public key $\pk_i$ as well as all helper values.
After the next updated commitment $\keyCom'$ is published, the user asks for an inclusion proof of $\pk_i$ and verifies it with \VerifyPK{}.
The inclusion proof here uses the tree construction from \cref{sec:secretTree}.
To this end, the helper values must also include the $\log(n)$ necessary $\sk_i \cdot \lagProof_{i,j}$ terms as in \cref{sec:secretTree}.

The provider constructs a proof that the change from $\keyCom$ to $\keyCom'$ is append-only. Let $E$ be the current epoch and let $J$ be the set of newly registered indices.
First, we ensure that the indices in $J$ follow a specific order, namely the typical FFT order of roots of unity in $\domain$.
Then, throughout the epoch, the provider builds up the commitment $\keyCom_E$ to $\keyPoly_E(x) = \sum_{i \in J} \sk_i \cdot \ell_i(x)$ which interpolates all the new keys.
Using the efficient \AMTMaintain{} routine, the provider also maintains an opening proof tree for $\keyCom_E$.
Suppose that at the start of epoch $E$, the key table commitment was $\keyCom$, that the last added user in $\keyCom$ was at index $c$ and that the next index after $c$ in FFT order was $k$.
This means that $\keyCom$ is {empty} from index $k$ onwards while $\keyCom_E$ is {empty} at all indices up to and including $c$.
Therefore, following \cref{sec:append}, we can use the respective proof tree paths for $\keyPoly(\omega^k)$ and $\keyPoly_E(\omega^c)$ to obtain append-only proofs. Specifically, when epoch $E$ ends, the provider outputs said proof tree paths (\Publish{}), which auditors will check during \VerifyKeys{} (in turn using \VerifyAppend{} as defined in \cref{sec:append}). Then, $\keyCom' = \keyCom + \keyCom_E$ must be an append-only evolution of $\keyCom$.
To prevent pathological attacks, the user verifies that the database entry at index $i$ is 0 after registration.

\subheading{Ensuring privacy.}
To ensure that all commitments and proofs hide the content of $\dbCom$, we employ an information-theoretic masking technique \cite{kzg}. Specifically, a second SRS with the same trapdoor $\tau$ but with base points $h \neq g, \hat{h} = \log_g(h) \cdot \hat{g}$ is assumed. This SRS is used for a masking polynomial $\maskPoly$ to hide $\dbPoly$, replacing $\dbCom$ with $\dbCom^* = \dbPoly(\tau) \cdot \hat{g} + \maskPoly(\tau) \cdot \hat{h}$.
The expression $\AMTMaintain(\dbCom^*, i, (\delta_i, \epsilon_i))$ is then understood as operating on the proof tree for $\dbCom^*$ and adding  $\delta_i \cdot [\lag_i(\tau) \cdot g] + \epsilon_i \cdot [\lag_i(\tau) \cdot h]$ to the nodes on the $i$th path (cf. \cref{sec:kzg}).
We must also ensure that the commitments to $\sigPoly$ and each $\sigPoly_E$ are hiding.
To this end, each time a user signs an update $\delta_i$, the user will simultaneously sign a re-randomization $\epsilon_i$ for the corresponding value in $\maskPoly$.
This way, $\sigCom_E$ is replaced with $\sum_i \sk_i \cdot \lag_i(\tau) \cdot (\delta_i \cdot g + \epsilon_i \cdot h)$ and $\sigCom$ is replaced with $\sum_i \sk_i \cdot \lag_i(\tau) \cdot (v_i \cdot g + \maskPoly_i \cdot h)$.
The lookup proofs and global proofs will be made zero-knowledge by also incorporating the randomizers $\maskPoly_i$ in them.
Note that the provider indeed knows the $\maskPoly_i$ and each of their re-randomizations $\epsilon_i$.
Thus, users only provide extra helper values with base point $h$ for the provider to be able to construct the right proofs.
We write these helpers as $\widetilde{\zeroQuoCom_{i,j}} = \sk_i \cdot \frac{\lag_i(\tau) \cdot \lag_j(\tau)}{\tau^n-1} \cdot h$.

\subsection{Security Analysis}\label{sec:analysis}

We now analyze the security of the \ac{prim} scheme in \cref{fig:prim}.
The full proofs are in \cref{sec:proofs}.

We start by showing that $\keyCom$ is, in fact, append-only due to \VerifyAppend{} as explained in \cref{sec:append}.
This property is crucial as it enables users to verify the inclusion of their key in $\keyCom$ \emph{once} and consider $\keyCom$ as trusted from then on.

\begin{prop}\label{prop:regSec}\label{prop:appendOnly}
If $\VerifyKeys(\keyCom, E, \keyCom', \keyProof)=1$, $\keyCom'$ is append-only from $\keyCom$.
\end{prop}

Next, we have a lemma that states that a simplified version of $(\KeyGen{}, \Sign{}, \VerifySig{})$ is  an existentially unforgeable signature.
The simplified scheme can be seen as an extension of BLS \cite{bls04} or as a special case of Boneh et al.'s \cite{homBLS} or the scheme from \cite{ap19}.
For further details, see \cref{sec:proofsPrelims}.

\begin{prop}\label{prop:uf}
The signature that signs $(E, \delta)$ as $\sk \cdot (\groupHash(E) + \delta \cdot \hat{g})$ is unforgeable.
\end{prop}

With \cref{prop:uf}, we can prove that our \ac{prim} satisfies a single-round version of data security.
I.e., when the public keys in $\keyCom$ are fixed and if the adversary is restricted to carrying out the data corruption in a single epoch, then the existence of a successful adversary would contradict \cref{prop:uf}.
For the full proof, which relies on the APK proof, AMT and zerocheck, see \cref{sec:proofSingleRound}.

\begin{prop}\label{prop:linHom}\label{prop:singleRound}
There is no single-round attacker (formalized in \cref{sec:proofSingleRound}) against the \ac{prim} in \cref{fig:prim}.
\end{prop}

Given \cref{prop:regSec,prop:linHom}, we show that the \ac{prim} is secure. Recall that an adversary breaks data security when creating \emph{any} length-$K$ history of commitments that auditors would accept where an honest user's value is corrupted.
We rule out such attacks by induction over $K$. The base case uses \cref{prop:linHom} while the induction step uses \cref{prop:appendOnly}.
The full proof is in \cref{sec:proofDbSec}.

\begin{prop}[Data Security]\label{prop:dbSec}\label{prop:dbSecInf}
The \ac{prim} in \cref{fig:prim} has data security as per \cref{def:dbSec}.
\end{prop}

Finally, we argue that the \ac{prim} scheme satisfies the data privacy definition by leveraging the fact that each user $i$'s balance updates $\delta_i$ includes a re-randomization factor $\epsilon_i$.
Since every public value that depends on $\delta_i$ is blinded by $\epsilon_i$, we show that there is no probabilistic polynomial-time adversary that can distinguish with a non-negligible probability between a genuine provider and a simulator in the sense of \cref{def:dbPriv} that replaces honest users' values with random values.
For the full proof, see \cref{sec:proofDbPriv}.

\begin{prop}[Data Privacy]\label{prop:dbPriv}
The \ac{prim} in \cref{fig:prim} has data privacy as per \cref{def:dbPriv}.
\end{prop}

\begin{figure*}[!t]
\centering
\begin{pchstack}[space=0.35\baselineskip]
\begin{pchstack}[boxed,space=0.25\baselineskip]
\procedure[linenumbering,skipfirstln,bodylinesep=\belowLine]{$\UpdateDB(i, \delta_i, (\sig_i, \epsilon_i))$}{
    \pcCommentLine{omitted: steps from \primUpdateDB{}} \\ \totalCom = \totalCom + \delta_i \cdot \hat{g} + \epsilon_i \cdot \hat{h} \\ \polQuoCom = \polQuoCom + \delta_i \cdot \left[\frac{\lag_i(\tau)-1/n}{\tau} \cdot g\right]
}
\procedure[linenumbering,skipfirstln,lnstart=2,bodylinesep=\belowLine]{$\EndEpoch$}{
    \pcCommentLine{omitted: steps from \primEndEpoch{}} \\
    \rangeProof = \RangeProve(\dbCom^*) \\
    \sumProof = \left(\totalCom, \polQuoCom\right) \\
    \pcreturn (\dbProof, \rangeProof, \sumProof) \text{ in place of } \dbProof
}
\end{pchstack}
\begin{pchstack}[boxed,space=0.25\baselineskip]
\procedure[linenumbering,skipfirstln,lnstart=5,bodylinesep=\belowLine]{$\VerifyDB(\keyCom, E, \dbCom^*, (\dbProof, \rangeProof, \sumProof); \sigCom)$}{
    \pcCommentLine{omitted: steps from \primVerifyDB{}} \\
    \pccheck \RangeVerify(\dbCom^*, \rangeProof) \\
    \pcparse \sumProof \text{ as } \left( \totalCom, \polQuoCom \right) \\
    \pccheck e\left(g , \dbCom^* - \totalCom/n \right) = e\left( \polQuoCom, \tau \cdot \hat{g} \right)
}
\end{pchstack}
\end{pchstack}
\vspace{-0.5 em}
\caption{Extending the \ac{prim} from \cref{fig:prim} to a \ac{model}.
\label{fig:ppolShort}}
\end{figure*}

\section{PPoL Design \& Evaluation}\label{sec:pol}

\subsection{\ac{model} Protocol Details}\label{sec:ppolConstruct}

As per \cref{sec:newModel}, we can now modularly build a \ac{model} on top of the \ac{prim} from \cref{sec:primConstruct} by integrating a summation proof and a range proof.
Here, the sum proof authenticates a commitment to a claimed sum of total liabilities, while the range proof ensures that no numerical overflows are possible during the summation.
We summarize the required modifications (\cref{fig:ppolShort}) and, due to space constraints, defer further details to \cref{sec:ppolDetails}.

The main idea for the summation proof is to use a KZG opening \cite{kzg} of $\dbCom$ at 0 to authenticate the value committed in a group element $\Total$ as $\total = \sum_{i \in [0,n-1]} \dbPoly(\omega^i)$ \cite{aurora} \cite[][Fact 1]{vampire}.
The proof is a commitment to $\polQuo$ with $\dbPoly - \total/n = \polQuo \cdot x$.
Importantly, a change $\delta_i$ to a single entry of $\dbPoly$ only requires updating $\polQuo$ with $\delta_i \cdot \frac{\lag_i - 1/n}{x}$.
For the range proof---black-boxed as $(\RangeProve, \RangeVerify)$ in \cref{fig:ppolShort}---we resort to a standard idea of expressing $\dbPoly$ bit by bit \cite{provisions}.
Here, it is assumed that balance values are always below a power of two and at most $|\FF|/n$ (where $|\FF| \approx 2^{255}$ in the KZG setting).
Let $m \leq \lfloor\log(|\FF|/n)\rfloor$.
Then, the range proof consists of commitments to $m$ binary-valued polynomials $\polBin_0, \dots, \polBin_{m-1}$, with accompanying binarity proofs, such that $\sum_{i \in [0,m-1]} \polBin_i \cdot 2^i = \dbPoly$.
That is, $\polBin_i$ interpolates the $i$th bits of the values in $\dbPoly$.
The binarity proofs are made zero-knowledge like in \cite[Appendix B]{libert23}.
Namely, the commitments are blinded by random multiples of the vanishing polynomial of $\domain$, yielding $\polBinCom_i = (\polBin_i(\tau) + \mu_i \cdot (\tau^n-1)) \cdot g$.
Binarity proofs are then constructed as per \cref{sec:quotients}, however, taking the randomization into account.
Finally, we batch the $m$ binarity proofs via random linear combination (see \cref{sec:ppolDetails}).
Note that since \emph{only} users can randomize their proof tree paths and some paths overlap, it is possible for some users to infer others' activity pattern.
A simple mitigation is to assign each user a sibling account, controlled by the provider or another semi-trusted party, which performs regular re-randomizations.
In this way, shared path segments between siblings obscure user activity without affecting balances or reducing liabilities.
Additionally, since the APK proof’s commitments ($\binCom$, $\binQuoCom$) leak activity, we apply the same zero-knowledge masking used in the range proof’s binarity checks (\cref{sec:ppolDetails}).

\begin{figure*}[t]
\centering
\includegraphics[width=0.75\textwidth]{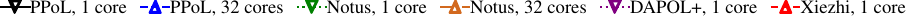}
\vspace{0.4em}

\subcaptionbox{Global proof creation.\label{fig:evalGlobalProver}}{
\includegraphics[width=0.26\textwidth]{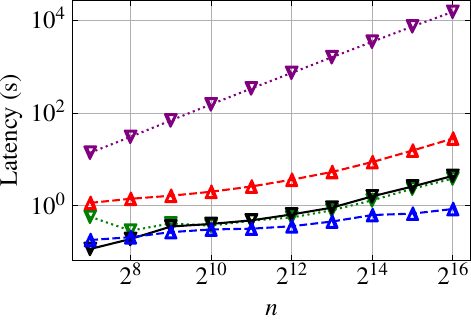}
\vspace{-0.5em}
}
\subcaptionbox{Inclusion proof creation.\label{fig:evalSingleProver}}{
\includegraphics[width=0.26\textwidth]{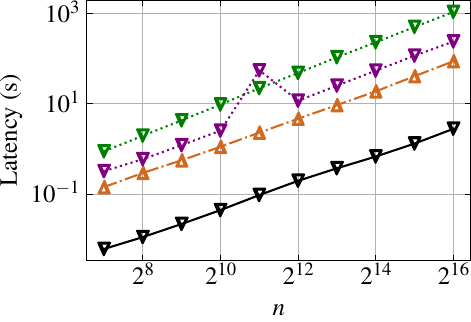}
\vspace{-0.5em}
}
\subcaptionbox{Global proof verification.\label{fig:evalGlobalVerify}}{
\includegraphics[width=0.26\textwidth]{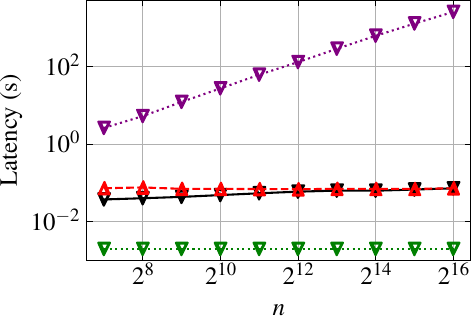}
\vspace{-0.5em}
}
\vspace{0.4em}

\subcaptionbox{Inclusion verification.\label{fig:evalSingleVerify}}{
\includegraphics[width=0.26\textwidth]{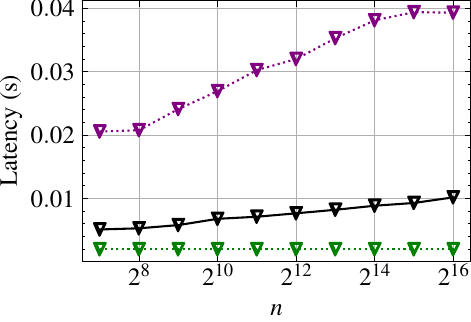}
\vspace{-0.5em}
}
\subcaptionbox{Provider throughput, 1 core.\label{fig:evalThroughput1}}{
\includegraphics[width=0.26\textwidth]{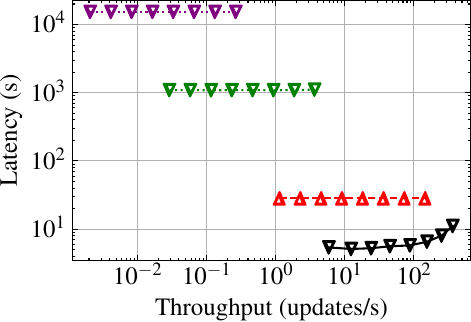} \vspace{-0.5em}
}
\subcaptionbox{Throughput, 32 cores.\label{fig:evalThroughput2}}{
\includegraphics[width=0.26\textwidth]{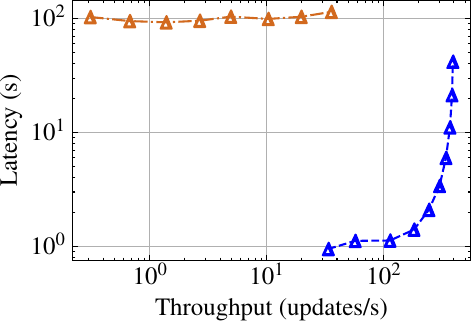} \vspace{-0.5em}
}
\vspace{-0.5 em}
\caption{Results.
In Figs. (a)--(d), if 1-core and 32-core values were equal, only one is shown.
In Figs. (e) and (f), $n = 2^{16}$.
Note: as Xiezhi's code base does not implement inclusion proofs, the measured performance is inherently overly optimistic.
\label{fig:eval}}
\end{figure*}

\subsection{Evaluation}\label{sec:eval}

We implemented our \ac{model} in Go, using gnark-crypto \cite{gnarkCrypto} for elliptic curve and finite field operations.\footnote{Our code is publicly available at~\cite{githubRepo}.}  We chose the curve BN254 and used bit width $m=64$ for the \ac{model}'s range proofs.

\subheading{Methodology.}
We first evaluate our implementation on \ac{model}-specific criteria and then report on two sets of experiments that compare our implementation against state-of-the-art PoLs.
We compare against DAPOL+ \cite{dapolplus}, Notus \cite{notus} and the Xiezhi scheme from a recent preprint \cite{xiezhi}, which also appears to have state-of-the-art performance.
All experiments are performed on a dual-64-core AMD EPYC 7742 machine with 1024 GB of RAM. Unless stated otherwise, the experiments ran single-threaded, were repeated 5 times and averaged.
For DAPOL+ and Xiezhi, no multi-threaded code was available.

The \emph{\ac{model}-specific experiments} measure the overhead for user registration.

Then, we start by \textbf{varying the total number of users $n$} while keeping the fraction of users who update within an epoch constant.
Similar to \cite[cf.][Section 6.3]{notus}, this fraction is set to $2^{-6}$ so that $|I| = n \cdot 2^{-6}$. Since most schemes allow separating the prover's work into a component for computing users' individual proofs and another component for the global auditor proofs, we evaluate these tasks separately.
In a real execution, the total proving time would be the sum of both tasks' timings.
Recall from \cref{sec:existing} that DAPOL+ is a PoL without a succinctly verifiable sum.
Therefore, the only way for an auditor to check a committed sum of liabilities against the database entries is iterating over all users' individual proofs.
This is why the DAPOL+ \enquote{global proof} creation and verification timings will be given as the sum of all $n$ individual users' timings.
We note that Xiezhi's code does not implement inclusion proofs.

Next, we vary \textbf{the number of users $I$ who update their balance}, while keeping $n$ constant. Here, we are interested in the latency and throughput witnessed by the server.
Latency means the time needed to compute all proofs published at the end of an epoch (lower bounding the epoch length).
Throughput is the number of balance updates the server can perform per unit of proving time.

\subheading{Evaluation results.}
In \ac{model},
upon registration, each user has a one-time $O(n)$ key generation and the provider must update $O(n)$ helper values. We find that, with $n=2^{16}$ users, this amounts to 28 seconds of local computation for each user and 420 (18) milliseconds per user for the single-threaded (32-core) provider.

The microbenchmarks for various values of $n$ are shown in \cref{fig:evalGlobalProver,fig:evalSingleProver,fig:evalGlobalVerify,fig:evalSingleVerify}.
We find that the \emph{global proof} creation takes comparable amounts of time (4 seconds at $2^{16}$) in single-threaded \ac{model}, single-threaded Notus and multi-threaded Notus.
Multi-threaded \ac{model} appears to benefit heavily from parallelism so that the timing is only 1 second at $n=2^{16}$.
This is likely due to our construct being bottlenecked only by size-$n$ multi-scalar multiplications and FFTs, which can be parallelized.
As shown in \cref{fig:evalGlobalProver}, both DAPOL+ and Xiezhi perform worse by up to an order of magnitude.
The creation of \emph{per-user inclusion proofs} takes 3.7 seconds for single-threaded and multi-threaded \ac{model}.
In contrast, single-threaded and multi-threaded Notus and DAPOL+ require 88, 1000 and 240 seconds, respectively, for $n=2^{16}$.
This behavior is likely due to both Notus and DAPOL+ having to recompute \emph{all} $n$ users' proofs from scratch while \ac{model} performs $O(\log(n))$-complexity tree updates \emph{only for those users who actually update}.
The verification timings of \ac{model}, Xiezhi and Notus are generally below one second, both for users' inclusion proofs and for the auditors' global proofs.
While in Notus, both verifications correspond to a Groth16 SNARK (2 milliseconds, irrespective of $n$), our \ac{model}'s inclusion proofs took 10 milliseconds and its and Xiezhi's global proofs took 70 milliseconds to verify at $n=2^{16}$.
We expect these differences to be insignificant in most settings.
The DAPOL+ global proof verifier is an outlier, at 40 minutes, as it is linear in $n$.

Provider-side latency and throughput for varying $|I|$ are shown in \cref{fig:evalThroughput1,fig:evalThroughput2}.
Our \ac{model}'s latency stays below 10 seconds even for as many as 300 balance updates per second, single-threaded.
Xiezhi's throughput seems similar (30 seconds of latency at 180 updates per second), however, the numbers are only a \emph{lower bound} as the latest Xiezhi code did not provide inclusion proofs.
For single-threaded Notus and DAPOL+, latency is one and, respectively, two orders of magnitude higher at throughputs which are one and, respectively, two orders of magnitude lower than PPoL.
For Notus, multi-threading improves latency---as expected from \cref{fig:evalSingleProver}---but not below 100 seconds at 35 updates per second.
The multi-threaded \ac{model} resembles the single-threaded case since for greater $|I|$, the less parallelizable per-update work dominates the global work.

\begin{rem}[Sharding]
In practice, PoLs use sharded architectures where users are divided into subsets and one PoL instance is run per subset \cite{binanceImpl,notus}.
We infer that, e.g., using our experiments' $n=2^{16}$ setup on 4096 machines handles $2^{28}$ users with the same relative latency and throughput gains over our baselines.
\end{rem}

\section{Conclusion}
In this work, we introduced a novel PoL model called Permissioned PoL (PPoL) that---unlike existing approaches---eliminates the need for user cooperation in detecting potential misbehavior by a malicious provider. Our model inherently prevents the provider from tampering with users' balance values and removes the burden on users to continuously verify their balances for potential manipulation.
We implemented our proposals and evaluated their performance when compared to state-of-the-art PoL proposals. Our evaluation results indicate that the strong security guarantees of our \ac{model} do not come with a prohibitive performance penalty.

\begin{acks}
The authors would like to thank Kyle Rudnick and Julian Willingmann for help in gathering data for \cref{fig:workload} and in performing the evaluation.
This work has been funded by the Deutsche Forschungsgemeinschaft (DFG, German Research Foundation) under Germany’s Excellence Strategy - EXC 2092 CASA - 390781972.
\end{acks}

\printbibliography

\appendix

\section{Proofs}\label{sec:proofs}
\subsection{Preliminaries}\label{sec:proofsPrelims}
\subheading{Hardness assumptions.}
Let $\FF$ be a prime-order finite field, $\GG, \hat{\GG}, \GG_T$ be groups of order $|\FF|$ and let $e : \GG \times \hat{\GG} \rightarrow \GG_T$ be a bilinear map.
Let $\secpar$ be a security parameter. Then, under the co-computational Diffie-Hellman (co-CDH) assumption, for every probabilistic polynomial-time algorithm, the probability to output $\beta \cdot \beta' \cdot \hat{g}$ when given $g, \beta \cdot g \in \GG, \hat{g}, \beta' \cdot \hat{g}, \beta \cdot \hat{g}$ as input is bounded by a negligible function in $\secpar$.
Further, under the $n$-Strong Diffie-Hellman ($n$-SDH) assumption, for every probabilistic polynomial-time algorithm, the probability to output $(\beta'', (\tau + \beta'')^{-1} \cdot g)$ when given $\left(\tau^i \cdot g, \tau^i \cdot \hat{g}\right)_{i \in [0,n]}$ as input is bounded by a negligible function in $\secpar$.

\subheading{Algebraic group model.}
Our security arguments rely on the algebraic group model (AGM) from \cite{agm} as an underlying paradigm.
In this model, it is presumed that from any algorithm that outputs elements in cryptographic groups, it is possible to extract representations of these group elements in terms of other known group elements that were given to the algorithm previously (like public generators).

We also rely on the following proposition which follows from both Boneh et al. \cite{homBLS} and from Aranha and Pagnin \cite{ap19}.
Indeed, it is a (one-dimensional, single-signer, trivial) special case of both Boneh et al.'s and Aranha and Pagnin's general (multi-dimensional) schemes' security theorems.
Because the proposition is a straightforward corollary of \cite[Theorem 6]{homBLS} or \cite[Theorem 1]{ap19}, we omit a proof here and refer to the cited works for details (with a warning about differing notations).
The claim can also be seen as a slight extension of the standard security statement for BLS signatures \cite{bls04}, so \cite{bls04} or \cite[Theorem 6.1]{agm} may also serve as references for proof.

\subheading{\cref{prop:uf}.}
In the following security experiment $\gameUf$, every probabilistic polynomial-time adversary has at most negligible success probability under the co-CDH assumption, modeling $\groupHash : \{0,1\}^* \rightarrow \hat{\GG}$ as a random oracle:
\begin{enumerate}
\item A prime-order finite field $\FF$, groups $\GG, \hat{\GG}, \GG_T$ of order $|\FF|$, a bilinear map $e : \GG \times \hat{\GG} \rightarrow \GG_T$, fixed elements $g \in \GG, \hat{g} \in \hat{\GG}$ and $\beta \in \FF$ are chosen.
\item $\pk = \beta \cdot g$ and the above parameters, except for $\beta$, are sent to the adversary.
\item The adversary gets oracle access to $\groupHash$.
\item The adversary gets access to a signing oracle that, on input $(j, \delta^{(j)})$, returns $\sig = \beta \cdot (\groupHash(j) + \delta^{(j)} \cdot \hat{g})$---but at most one response per $j$.
\item Finally, the adversary outputs $E \in \{0,1\}^*, z\in \FF, \sig \in \hat{\GG}$.
\item The adversary wins if
\begin{itemize}
\item $e(\sig, \hat{g}) = e(\pk, \groupHash(E) + z \cdot \hat{g})$ and
\item either
    \begin{itemize}
        \item no signing query $(E, \cdot)$ was made or
        \item $z \neq \delta^{(E)}$.
    \end{itemize}
\end{itemize}
\end{enumerate}

\subsection{Proposition \ref{prop:linHom}}\label{sec:proofSingleRound}

\Cref{prop:linHom} is a stepping stone towards the main data security theorem.
To state \cref{prop:linHom} precisely, we first define an auxiliary security game $\gameSingleRound{}$ for \ac{prim} schemes.
Like in the data security game ($\gameDbSec$), the attacker's goal is to corrupt an honest user's value.
But here, only a single epoch change is considered, the table of keys is assumed to be static.
$\gameSingleRound{}$ proceeds as follows:
\begin{enumerate}
\item Public parameters \pp{} and a key pair $\sk_i, \pk_i$ are generated by $\Setup(\secpar, n)$ and, respectively, $\KeyGen(i)$.
\item \pp{} and $\pk_i$ are sent to the adversary.
\item The adversary gets access to a \Sign{} oracle for $\sk_i$.
On input $(j, \delta^{(j)})$, the oracle outputs $\sig^{(j)}$ as computed by $\Sign(\sk_i, j, \delta^{(j)})$.
However, if previously there already was another query of the form $(j, \cdot)$, then the same output as before is returned.
\item Eventually, the adversary outputs $E,\allowbreak{} \keyCom,\allowbreak{} \dbCom,\allowbreak{} \dbProof,\allowbreak{} \dbCom',\allowbreak{} \dbProof',\allowbreak{} \keyProof_i,\allowbreak{} \dbProof_i,\allowbreak{} \dbProof_i',\allowbreak{} z,\allowbreak{} z'$ and a starting state $\aux$ for the auditor.
\item The adversary wins if:
\begin{itemize}
\item $\VerifyPK(\keyCom, i, \pk_i, \keyProof_i) = 1$.
\item $\VerifyDB(\keyCom, E-1, \dbCom, \dbProof; \aux) = 1$ where, during execution, $\aux$ is updated to $\aux'$.
\item $\VerifyDB(\keyCom, E, \dbCom', \dbProof'; \aux') = 1$.
\item $\VerifyLU(\dbCom, i, z, \dbProof_i) = 1$.
\item $\VerifyLU(\dbCom', i, z', \dbProof_i') = 1$.
\item $z'-z \neq \delta^{(E)}$ where, if no signing query of the form $(E, \cdot)$ was made, $\delta^{(E)}=0$.
\end{itemize}
\end{enumerate}

\subheading{Proof.}
We prove that our \ac{prim} scheme (\cref{fig:prim}) ensures that, for any $n$ and $i$, a malicious provider has at most a negligible chance of winning $\gameSingleRound{}$.
Suppose there was an attacker $\advA{}$ with non-negligible success probability.
In the AGM under the $n$-SDH assumption and when $\groupHash$ is modeled as a random oracle, it would then be possible to construct an attacker $\advB{}$ against $\gameUf{}$ (\cref{sec:proofsPrelims}) with a similar success probability---in contradiction of \cref{prop:uf}.

To begin, $\advB{}$ is given the description of $\GG, \hat{\GG}, \GG_T$ and $\beta \cdot g$.
In order to simulate the game $\gameSingleRound{}$ for $\advA$, $\advB$ then runs \Setup{} and remembers the trapdoor $\tau$ as well as $\eta$ such that $h = \eta \cdot g$.
The public key $\pk_i^*$ for user $i$ is generated with $\beta \cdot g$ taking the place of $\pk_i$.
The remaining components for $\pk_i^*$ are correctly computed using knowledge of $\tau$, e.g., $\zeroQuoCom_{i,j} = \frac{\lag_i(\tau) \cdot \lag_j(\tau)}{\tau^n-1} \cdot [\beta \cdot g]$.
After $\pp$ and $\pk_i^*$ are sent to $\advA$, signing queries and random oracle queries are simulated as follows.
When $\advA$ makes a random oracle query for $j$, $\advB$ makes the same query to $\groupHash$ and returns the result.
When $\advA$ makes a signing query for $\left(j, (\delta^{(j)}, \epsilon^{(j)})\right)$, $\advB$ makes a signing query for $\left(j, (\delta^{(j)} + \eta \cdot \epsilon^{(j)}) \cdot \lag_i(\tau)\right)$ and returns the result to $\advA$.
Note that when $\advA$ outputs, $z$ and $z'$ will both come with two components (the second one of which is related to our scheme's randomization), $z=(z_1, z_2), z'=(z_1',z_2')$.
Simplifying notation, let us re-define $z = z_1 + \eta \cdot z_2$ and $z' = z_1' + \eta \cdot z_2'$.
Now, after $\advA$ outputs and the checks pass, we can claim (always except for negligible errors):
\begin{itemize}
\item By the soundness of the AMT verification identity (from \AMTVerify{} during \VerifyPK{}), $\keyCom = \keyPoly(\tau) \cdot g$ for $\keyPoly(x)$ with $\keyPoly(\omega^i) = \beta$.
\item By the soundness of the AMT verification identity (from \AMTVerify{} during \VerifyLU{}), $\dbCom = \dbPoly(\tau) \cdot \hat{g}$ for $\dbPoly(x)$ with $\dbPoly(\omega^i) = z$.
\item By the soundness of the AMT verification identity (from \AMTVerify{} during \VerifyLU{}), $\dbCom' = \dbPoly'(\tau) \cdot \hat{g}$ for $\dbPoly'(x)$ with $\dbPoly'(\omega^i) = z'$.
\item By the successful zerocheck\footnote{In the AGM, the adversary outputting $\zeroQuoCom$ such that $e(\keyCom, \dbCom) = e(\sigCom, \hat{g}) + e(\zeroQuoCom, \tau^n \cdot \hat{g} - \hat{g})$ implies the extractability of a quotient polynomial $\zeroQuo$ that satisfies the zerocheck identity from \cref{sec:quotients}. For a proof from $n$-SDH of the analogous statement about the sumcheck identity, see, for example, \cite[Lemma 5.6]{das23}.\label{ft:zerocheck}} during the first execution of \VerifyDB{}, $\sigCom := \aux + \sigCom_{E-1}$ (where $\sigCom_{E-1}$ is part of $\dbProof$) commits to $\sigPoly(x)$ with $\sigPoly(\omega^i) = \beta \cdot z$.
Further, $\aux' = \sigCom$.
\item By the successful zerocheck during the second execution of \VerifyDB{}, $\sigCom' := \sigCom + \sigCom_E$ (where $\sigCom_{E}$ is part of $\dbProof'$) commits to $\sigPoly'(x)$ with $\sigPoly'(\omega^i) = \beta \cdot z'$.
\end{itemize}

With the previous claims, by additive homomorphism, we have that $\sigCom_E = \sigCom' - \sigCom = \sigPoly_E(\tau) \cdot g$ for $\sigPoly_E(x)$ with
\begin{align}
\sigPoly_E(\omega^i) &= \beta \cdot (z'-z) \nonumber\\
&= \beta \cdot ((z'_1 - z_1) + \eta \cdot (z'_2-z_2)). \label{eq:fe}
\end{align}
Further:
\begin{itemize}
\item By the soundness of the APK verification identities (from \APKVerify{}), $\apk_E = \sum_{j \in I} \keyPoly(\omega^j)$ (where $\apk_E$ is part of $\dbProof$ and $I \subseteq [0,n-1], i \in I$).
\item By the signature check during the second execution of \VerifyDB{},
\begin{align}\label{eq:sig}
\sig_E = \log_g(\apk_E) \cdot \groupHash(E) + \sigPoly_E(\tau) \cdot \hat{g}.
\end{align}
\end{itemize}

Next, $\advB$ uses the fact that $\advA$ is algebraic to express $\sigPoly_E$ and $\keyCom$ as combinations of other known group elements.
Known relations between group elements allow simplifying the resulting expressions.
In particular, we collect all terms involving $\beta$ and separate them from those without $\beta$.
$\advB$ thus obtains coefficients of $\gamma(x), \zeta(x), \theta(x), \kappa(x)$ such that:
\begin{align*}
\sigCom_E &= \sigPoly_E(\tau) \cdot g = (\gamma(\tau) + \beta \cdot \zeta(\tau)) \cdot {g} \\
\keyCom &= \keyPoly(\tau) \cdot g = (\theta(\tau) + \beta \cdot \kappa(\tau)) \cdot g.
\end{align*}
This also means, except for negligible probability over the uniform choice of $\tau \in \FF$, that
\begin{align*}
\sigPoly_E = \gamma + \beta \cdot \zeta \qquad
\keyPoly = \theta + \beta \cdot \kappa
\end{align*}
as polynomials and, due to \cref{eq:fe},
\begin{align*}
\gamma(\omega^i) + \beta \cdot \zeta(\omega^i) &= \beta \cdot (z'-z).
\end{align*}
Since we collected all terms involving $\beta$ (and therefore all involving $\beta \cdot \lag_i(\tau)$) into $\zeta$, we can assume $\gamma(\omega^i) = 0$ and $\zeta(\omega^i) = z'-z$.
We decompose $\zeta$ as
\begin{align*}
\zeta(x) = \zeta^*(x) + (z'-z) \cdot \lag_i(x)
\end{align*}
where $\zeta^*(\omega^i) = 0$.
So far, \cref{eq:sig} can be rewritten as
\begin{align}\label{eq:sig1}
\begin{split}
\sig_E &= \left(\sum_{j \in I} \theta(\omega^j) + \beta \cdot \kappa(\omega^j)\right) \cdot \groupHash(E) + (\gamma(\tau) + \beta \cdot \zeta(\tau)) \cdot \hat{g}.
\end{split}
\end{align}
Let $\sig$ be the queried signature for $\left(E, (\delta^{(E)}, \epsilon^{(E)})\right)$ such that $(z_1'-z_1, z_2'-z_2) \neq (\delta^{(E)}, \epsilon^{(E)})$.
Let $\delta = \delta^{(E)} + \eta \cdot \epsilon^{(E)}$.
Thus,
\begin{align*}
\sig = \beta \cdot (\groupHash(E) + \delta \cdot \lag_i(\tau) \cdot \hat{g}).
\end{align*}

We next differentiate two cases and, in each one, show how to use $\sig$ in order to extract valid signatures for epoch $E$ and different messages than were queried---forgeries in the sense of $\gameUf{}$, contradicting \cref{prop:uf}.
$\advB$'s knowledge of $\tau$ ensures that $\advB$ can always compute the corresponding messages for which the forgeries are valid.
\begin{enumerate}
\item If $\sum_{j \in I} \kappa(\omega^j) \neq 0$:
consider \cref{eq:sig1} and subtract \[\left(\sum_{j \in I} \theta(\omega^j)\right) \cdot \groupHash(E) + \gamma(\tau) \cdot \hat{g}\] from both sides.
Divide by $\sum_{j \in I} \kappa(\omega^j)$.
The result is \[\sig' = \beta \cdot \left(\groupHash(E) + \frac{\zeta(\tau)}{\sum_{j \in I} \kappa(\omega^j)} \cdot \hat{g}\right).\]
When assuming $\zeta(\tau) \neq 0$, combining $\sig$ and $\sig'$, e.g., as $\mu \cdot \sig + (1-\mu) \cdot \sig'$ for any $\mu$, allows forging arbitrary signatures. Should, on the other hand, $\zeta(\tau) = 0$, $\advB$ can use $\sig + \mu \cdot (\sig-\sig')$ as a forgery.
\item If $\sum_{j \in I} \kappa(\omega^j) = 0$:
In this case, subtracting \[\left(\sum_{j \in I} \theta(\omega^j)\right) \cdot \groupHash(E) + \gamma(\tau) \cdot \hat{g}\] from both sides of \cref{eq:sig1} yields \[\sig' = \beta \cdot (\zeta^*(\tau) + (z'-z) \cdot \lag_i(\tau)) \cdot \hat{g}.\]
Arbitrary signatures can now be forged by combining $\sig$ and $\sig'$, for example, as $\sig - \sig'$.
Here, note that $\delta \neq z'-z$ helps to rule out pathological cases where forgeries could otherwise not be computed.
\end{enumerate}

As the success of $\advB$ in $\gameUf{}$ is the success probability of $\advA$ in $\gameSingleRound{}$ conditioned on events with overwhelming probabilities, $\advB$ has non-negligible success probability.

\subsection{\cref{prop:dbSec}}\label{sec:proofDbSec}

We show that the \ac{prim} scheme from \cref{fig:prim} is secure.
We proceed by induction over $K$.
The base case $K=1$ follows essentially the same strategy as the proof of \cref{prop:singleRound}.
The induction step uses the algebraic group model, \cref{prop:appendOnly} and the induction hypothesis and then also uses the steps of our proof of \cref{prop:singleRound}.

\subheading{Base case, $K=1$.}
Assuming an attacker $\advA$ against data security for $K=1$, we construct an attacker $\advB$ against $\gameUf$ (\cref{sec:proofsPrelims}) with similar success probability.
$\advB$ starts with $\pk$ and simulates $\gameDbSec$ for $\advA$ by choosing $n$, running $\Setup$ to generate $\pp$ (remembering the trapdoor $\tau$ and $\eta$ such that $h = \eta \cdot g$), choosing $i \in [n]$, setting $\pk_i = \pk$ and setting $\pk_i^*$ based on $\pk_i$ and knowledge of $\tau$ and $\eta$.
$\advA$'s random oracle requests are passed to the $\groupHash$ oracle.
$\advA$'s signing requests $(j, (\delta^{(j)}, \epsilon^{(j)}))$ are passed to the signing oracle as $\left(j, (\delta^{(j)} + \eta \cdot \epsilon^{(j)}) \cdot \lag_i(\tau)\right)$.
Eventually, $\advA$ outputs $\keyCom^{(0)}, \dbCom^{(0)}, \dbCom^{(1)}, \keyCom^{(1)}, \dbProof^{(1)}, \keyProof^{(1)}, \keyProof_i^{(0)}, \dbProof_i^{(1)}, z$ such that all verification checks pass, but: if a query was made for epoch 1, $z \neq (\delta^{(1)}, \epsilon^{(1)})$, else $z \neq (0,0)$.

From $\VerifyPK\left(\keyCom^{(0)}, i, \pk_i, \keyProof_i^{(0)}\right)$ accepting, we deduce that $\keyCom^{(0)}$ commits to $\keyPoly(x)$ with $\keyPoly(\omega^i) = \beta$.
From $\VerifyLU\left(\dbCom^{(1)}, i, z, \dbProof_i^{(1)}\right)$ accepting, we deduce that $\dbCom^{(1)}$ commits to $\dbPoly(x)$ with $\dbPoly(\omega^i) = z_1 + \eta \cdot z_2$.
From $\VerifyDB\left(\keyCom^{(0)}, 1, \dbCom^{(1)}, \dbProof^{(1)}; 0\right)$ accepting, we deduce that $\sigCom_1 \in \dbProof^{(1)}$ commits to $\sigPoly_1(x)$ with $\sigPoly_1(\omega^i) = \beta \cdot (z_1 + \eta \cdot z_2)$.
(The previous claims all hold except with negligible probability.)
But now we are in the same situation as in the proof of \cref{prop:singleRound}---i.e., we just derived an analogue of \cref{eq:fe}.
Hence, we can follow the same strategy to extract a forgery against $\gameUf$.

\subheading{Induction step, $K-1 \rightarrow K$.}
Suppose there was an attacker $\advA$ that would output a convincing corrupted history of length $K$.
In the algebraic group model, due to our ability to represent $\dbCom^{(K-1)}$ as a combination of $\tau^i$ terms, we can be confident that $\dbCom^{(K-1)}$ is a commitment to \emph{some} polynomial.
Therefore, there must exist a value $y$ and a corresponding proof $\dbProof_i^{(K-1)}$ that the $i$th entry in $\dbCom^{(K-1)}$ is $y$.
Importantly, the ability to represent $\dbCom^{(K-1)}$ also implies the ability to compute such a proof.
First consider the case $E<K-1$.
By the induction hypothesis, $y=\sum_{k \in \expList \cap [E+1,K-2]} \delta^{(k)} + \eta \cdot \epsilon^{(k)}$.
Take this together with the fact that all $\keyCom^{(k)}, k>E$ are successive append-only updates where $\pk_i$ is always contained at position $i$ (due to \VerifyKeys{} as per \cref{prop:appendOnly}).
But then, the reasoning used in the proof of \cref{prop:singleRound} applies to the values $\dbCom^{(K-1)}, \dbProof^{(K-1)}, \dbCom^{(K)}, \dbProof^{(K)}, \dbProof_i^{(K-1)}, \dbProof^{(K)}, y, z$.
This includes, critically, the derivation of \cref{eq:fe} and all subsequent steps which will lead to a forgery for $\gameUf$.

In case $E=K-1$, we use the fact that $\VerifyPK(\keyCom^{(K-1)},\allowbreak{} i,\allowbreak{} \pk_i,\allowbreak{} (\keyProof_i^*,\allowbreak{} \dbProof_i^{(K-1)});\allowbreak{} \dbCom^{(K-1)})\allowbreak{} = 1$ also implies $\VerifyLU(\dbCom^{(K-1)},\allowbreak{} i,\allowbreak{} (0,\allowbreak{}0),\allowbreak{} \dbProof_i^{(K-1)})\allowbreak{} = 1$.
As previously, we can deduce $\keyPoly(\omega^i) = \beta, \dbPoly(\omega^i) = 0, \sigPoly(\omega^i) = 0$ (where $\sigPoly$ is from the auxiliary input to \VerifyDB{} in round $K-1$).
Further, for $\sigPoly_K \in \dbProof^{(K)}$, we have $\sigPoly_K(\omega^i) = \beta \cdot (z_1 + \eta \cdot z_2) \neq \beta \cdot (\delta^{(K)} + \eta \cdot \epsilon^{(k)})$ as per the round-$K$ \VerifyDB{} check.
With the previous equation replacing \cref{eq:fe}, we run through the same reasoning as in the proof of \cref{prop:singleRound} towards a forgery for $\gameUf$.

\subsection{\cref{prop:dbPriv}}\label{sec:proofDbPriv}

We now argue that the use of randomizers within the \ac{prim} from \cref{fig:prim} allows for a straightforward simulator construction.
The simulator runs all of the provider's algorithms normally, however, always replacing honest users' values with dummy inputs:
\begin{itemize}
\item $\Sim_1$ runs \Setup{} and returns $\pp$.
\item $\Sim_2$, on input $i, \pk_i$, runs $\AddPK(i, \pk_i)$.
\item $\Sim_3$, on input $i$, calls $\KeyGen(i)$ to compute $\sk_i, \pk_i$, sets $\dbPoly_i, \maskPoly_i = 0$, calls $\AddPK(i, \pk_i)$ and returns $\pk_i$.
\item $\Sim_4$, on input $i, \delta, \sig$, runs $\UpdateDB(i, \delta, \sig)$.
\item $\Sim_5$, on input $i$, picks $\delta_i$ arbitrarily, computes $(\sig, \epsilon_i) = \Sign(\sk_i, E, \delta_i)$ and calls $\UpdateDB(i, \delta_i, (\sig_i, \epsilon_i))$.
\item $\Sim_6$ runs \EndEpoch{} and returns the output.
\end{itemize}

\begin{figure*}[!t]
\centering
\begin{pcvstack}[space=0.25\baselineskip]
\begin{pchstack}[boxed,space=0.35\baselineskip]
\procedure[linenumbering,skipfirstln,bodylinesep=\belowLine]{$\UpdateDB(i, \delta_i, (\sig_i, \epsilon_i))$}{
    \pcCommentLine{omitted: steps from \primUpdateDB{}} \\
    \dbPoly = \dbPoly + \delta_i \cdot \lag_i(x) \\
    \maskCom = \maskCom + \epsilon_i \cdot [\lag_i(\tau) \cdot \hat{g}] \pcskipln\\
    \pcCommentLine{update sum and sum proof} \\
    \totalCom = \totalCom + \delta_i \cdot \hat{g} + \epsilon_i \cdot \hat{h}  \\
    \polQuoCom = \polQuoCom + \delta_i \cdot \textstyle\left[\frac{\lag_i(\tau)-1/n}{\tau} \cdot g\right]
}
\begin{pchstack}[space=-0.085cm]
    \procedure[linenumbering,skipfirstln,lnstart=4,bodylinesep=\belowLine]{\EndEpoch{}}{
        \pcCommentLine{omitted: steps from \primEndEpoch{}} \\
        \text{find binary-valued } \polBin_0(x), \dots, \polBin_{m-1}(x) \pcskipln\\
        \qquad \text{ with } \textstyle\sum_{i \in [0,m-1]} \polBin_i(x) \cdot 2^i = \dbPoly(x) \\
        \pcfor i \in [0,m-1]: \\
        \quad \text{pick random } \mu_i, \hat{\mu_i} \\
        \quad \polBinCom_i = \mu_i \cdot \left[(\tau^n-1) \cdot h\right] + \polBin_i(\tau) \cdot g \\
        \quad \hat{\polBinCom_i} = \hat{\mu_i} \cdot \left[(\tau^n-1) \cdot \hat{g}\right] + \polBin_i(\tau) \cdot \hat{g} \\
        \gamma_1, \gamma_2 = \fieldHash(\polBinCom_1, \dots, \polBinCom_m) \\
        \label{line:rangeQuoCom} \rangeQuoCom = \textstyle\sum_{i \in [0,m-1]} \gamma_1^i \cdot \left( \frac{\polBin_i(\tau)^2-\polBin_i(\tau)}{\tau^n-1} \cdot g + \mu_i \cdot \polBin_i(\tau) \cdot h \right. \pcskipln\\
        \qquad\qquad\quad \left. + \hat{\mu}_i \cdot \polBin_i(\tau) \cdot g + \hat{\mu}_i \cdot \mu_i \cdot \left[(\tau^n-1) \cdot h\right] - \hat{\mu}_i \cdot g \vphantom{\frac{a}{b}}\right) \phantom{aaa} \pcskipln\\
        \qquad\qquad\quad + \gamma_2^i \cdot (\mu_i \cdot h - \hat{\mu_i} \cdot g) \\
        \text{pick random } \gamma_0
    }
    \procedure[linenumbering,lnstart=12,bodylinesep=\belowLine]{\phantom{abc}}{
        \label{line:rangeRecompStart} \maskCom^* = \maskCom + \gamma_0 \cdot \left[(\tau^n-1) \cdot \hat{g}\right] \\
        \label{line:rangeRecompEnd} M^* = - \left( \gamma_0 + \textstyle\sum_{i \in [0,m-1]} \mu_i \cdot 2^i \right) \cdot h \\
        \rangeProof = \left( \polBinCom_0, \dots, \polBinCom_{m-1}, \maskCom^*, M^*, \rangeQuoCom \right) \\
        \sumProof = \left(\totalCom, \polQuoCom\right) \pcskipln\\
        \pcCommentLine{APK proof changes} \\
        \label{line:apkChangesStart} \text{pick random } \mu, \hat{\mu} \\
        \binCom = \binCom + \mu \cdot [(\tau^n-1)] \cdot h \\
        \hat{\binCom} = \hat{\binCom} + \hat{\mu} \cdot [(\tau^n-1)] \cdot g \\
        \gamma_3 = \fieldHash'(\binCom, \hat{\binCom}) \\
        U = U + \mu \cdot \binPoly(\tau) \cdot h + \hat{\mu} \cdot \binPoly(\tau) \cdot g \pcskipln\\
        \qquad + \mu \cdot \hat{\mu} \cdot [(\tau^n-1) \cdot h] - \hat{\mu} \cdot g \pcskipln\\
        \qquad + \gamma_3 \cdot (\mu \cdot h - \hat{\mu} \cdot g) \\
        \label{line:apkChangesEnd} T = T + \hat{\mu} \cdot \keyCom \\
        \pcreturn (\dbProof, \rangeProof, \sumProof) \text{ in place of } \dbProof
    }
\end{pchstack}
\end{pchstack}
\begin{pchstack}[boxed,space=0.35\baselineskip]
\procedure[linenumbering,skipfirstln,lnstart=23,bodylinesep=\belowLine]{$\VerifyDB(\keyCom, E, \dbCom^*, (\dbProof, \rangeProof, \sumProof); \sigCom)$}{
    \pcCommentLine{omitted: steps from \primVerifyDB{} with modified \APKVerify{} for $\pi_{\apk}$} \\
    \pcparse \rangeProof \text{ as } \left( \polBinCom_0, \dots, \polBinCom_{m-1}, \maskCom^*, M^*, \rangeQuoCom \right) \\
    \gamma_1, \gamma_2 = \fieldHash(\polBinCom_1, \dots, \polBinCom_m) \\
    \label{line:rangeQuoComVerify} \pccheck e(\rangeQuoCom, \left[(\tau^n-1) \cdot \hat{g}\right]) = \textstyle\sum_{i \in [0,m-1]} e\left(\gamma_1^i \cdot (\polBinCom_i - g), \hat{\polBinCom}_i\right) + e\left(\textstyle\sum_{i \in [0,m-1]} \gamma_2^i \cdot \polBinCom_i, \hat{g}\right) - e\left(g, \textstyle\sum_i \gamma_2^i \cdot \hat{\polBinCom}_i\right) \\
    \label{line:rangeRecompVerify} \pccheck e(g, \dbCom^*) - e\left(\textstyle\sum_{i \in [0,m-1]} \polBinCom_i \cdot 2^i, \hat{g} \right) = e(h, \maskCom^*) + e\left(M^*, [(\tau^n-1) \cdot \hat{g}] \right) \\
    \pcparse \sumProof \text{ as } \left( \totalCom, \polQuoCom \right) \\
    \pccheck e\left(g , \dbCom^* - \totalCom/n \right) = e\left( \polQuoCom, \tau \cdot \hat{g} \right)
}
\end{pchstack}
\end{pcvstack}
\caption{\ac{prim}-based \ac{model} with unrolled range proof.
$\fieldHash,\fieldHash' : \{0,1\}^* \rightarrow \FF$ are hash functions.
\label{fig:ppolConstruct}}
\end{figure*}
To see why this simulator's output is indistinguishable from the real provider's output---even though the simulator has no access to honest user's values---, first note that if no \enquote{\textbf{honest value update}} action is performed at all, then the output is equal and, therefore, obviously indistinguishable.
Otherwise, we know that at least one honest value update has been performed over the course of the experiment.
We also know that honest users always sign freshly uniformly sampled re-randomization components $\epsilon_i$ along with each value update $\delta_i$.
From this and from the specification in \cref{fig:prim}, we can deduce that each element of the real provider's output that is modified---during $\UpdateDB(i, \delta_i, (\sig_i, \epsilon_i))$---in a way that depends on $\delta_i$, is, in fact, uniformly distributed over $\GG$ or $\hat{\GG}$.
But by our definition of $\Sim_5$, the corresponding elements in the simulator's output are distributed the same way.

\section{PPoL Details}\label{sec:ppolDetails}

\cref{fig:ppolConstruct} shows the modifications that turn the \ac{prim} from \cref{sec:primConstruct} into a \ac{model} in full detail.
At a high level, \UpdateDB{} implements the incremental computation of the sum proof (\cref{sec:ppolConstruct}), while the modified \EndEpoch{} implements the range proof and \VerifyDB{} is augmented to check both the sum proof as well as the range proof.
In \UpdateDB{}, besides the sum proof (line 4), we also make explicit the bookkeeping of the database interpolant $\dbPoly(x)$ (line 1), of the committed interpolant of randomizers $\maskCom$ (line 2) and of the committed sum of liabilities $\totalCom$ (line 3).

\subheading{Range proof.}
The main idea of the range proof follows \cite[Appendix B]{libert23} in that we commit to the bits of values in the database.
These bit commitments must undergo binarity proofs (\cref{sec:quotients}) and are, therefore, given in both groups.
To preserve privacy, the commitments are hidden with $\mu_i$-/$\hat{\mu_i}$-multiples of the vanishing polynomial $x^n-1$.
These added multiples must also be taken into account when verifying that two commitments in $\GG$ and $\hat{\GG}$ commit to the same bit polynomial.
For $\GG$-commitments, we add the scaled vanishing polynomial using $h$ since the $n$th power of $\tau$ is not available on $g$.

\subheading{Binarity and equivalence.}
For each bit polynomial $\polBin_i$, the required binarity identity (checked \enquote{in the exponent}) is now

\begin{align*}
& (\polBin_i + \mu_i \cdot (x^n-1) - 1) \cdot (\polBin_i + \hat{\mu}_i \cdot (x^n-1)) \\
=&\ \polBin_i^2 - \polBin_i + (\polBin_i \cdot (\mu_i + \hat{\mu}_i) \\
&\quad + \mu_i \cdot \hat{\mu}_i \cdot (x^n-1) - \hat{\mu}_i) \cdot (x^n-1).
\end{align*}

We know that $x^n-1$ divides $\polBin_i^2 - \polBin_i$ as $\polBin_i$ is binary-valued.
Thus, the quotient-based binarity proof for a privacy-preserving bit commitment is a commitment to
$$
(*) = \frac{\polBin_i^2 - \polBin_i}{x^n-1} + \polBin_i \cdot (\mu_i + \hat{\mu}_i) + \mu_i \cdot \hat{\mu}_i \cdot (x^n-1) - \hat{\mu}_i.
$$
Next, to ensure that a hidden $\GG$-commitment $\polBinCom_i$ and a hidden $\hat{\GG}$-commitment $\hat{\polBinCom_i}$ contain the same polynomial, we prove that their discrete logarithms only differ by some multiple of $x^n-1$.
Specifically, we use a commitment to
\[
(**) = \mu_i - \hat{\mu_i}
\]
that is compared to $e(\polBinCom_i, \hat{g})-e(g, \polBinCom_i)$ via the pairing.

The verification equations for the binarity proof and for the equivalence proof can be batched together because both perform a check against $x^n-1$.
That is, for randomly chosen $\gamma_1, \gamma_2$, the previous two checks are replaced with a single check against the combined proof
$\gamma_1 \cdot (*) + \gamma_2 \cdot (**).$
Using powers of $\gamma_1$ and $\gamma_2$, we extend this batching strategy to also combine the proofs for the $m$ different $\polBinCom_i$.
Overall, we have described the computation of $\rangeQuoCom$ in line \ref{line:rangeQuoCom} in \cref{fig:ppolConstruct}.
The corresponding verifier check is in line \ref{line:rangeQuoComVerify}.

\subheading{Dealing with $\maskPoly$.}
The range proof must authenticate $\sum_i 2^i \cdot \polBinCom_i$ against the database commitment.
Recall that if $\dbPoly$ interpolates all users' balance values, and $\maskPoly$ interpolates all users' (re-)randomizers, then the public database commitment is $\dbCom^* = \dbPoly(\tau) \cdot \hat{g} + \maskPoly(\tau) \cdot \hat{h}.$
But computing $\sum_i 2^i \cdot \polBinCom_i$ gives $\dbPoly(\tau) \cdot g + \sum_i 2^i \cdot \mu_i \cdot (\tau^n-1) \cdot h$.
Clearly,
\begin{align}\label{eq:polBitDecomp}
\log_g\left(\sum_i 2^i \cdot \polBinCom_i\right) \neq \log_{\hat{g}}\left(\dbCom^*\right),
\end{align}
which might seem like an impediment to the recomputation-based range proof.
But observe that the bit decomposition check can be modified to allow the left-hand side and right-hand side of \cref{eq:polBitDecomp} to differ (only) by \textbf{(1)} a multiple of $h$ and \textbf{(2)} a multiple of $x^n-1$.
Stated like this, the \enquote{allowed} multiple of $h$ could be set to $\maskCom = \maskPoly(\tau) \cdot \hat{g}$ and the multiple of $x^n-1$ could be computed based on $\sum_i 2^i \cdot \mu_i$.
But if the prover sent out $\maskCom$ directly, the values of $\maskPoly$ would not be hidden.
Thus, $\maskCom$ is blinded by $\gamma_0 \cdot (x^n-1)$ for a random $\gamma_0$.
This yields a range proof that is adjusted to $\maskPoly$ without revealing private values.
Direct calculation shows that $\maskCom^*$ and $M^*$ as computed in lines \ref{line:rangeRecompStart} and \ref{line:rangeRecompEnd} of \cref{fig:ppolConstruct} are the resulting allowed multiples.
The verifier check is shown in line \ref{line:rangeRecompVerify}.

\subheading{APK proof.}
The following changes are made to \EndEpoch{} and \APKVerify{} (in \VerifyDB{}) so the APK proof hides active users' indices.
As in the range proof, the APK proof's bit vector ($\binPoly$) is masked with a random multiple of the vanishing polynomial and an equivalence proof of the from $(\mu-\hat{\mu})\cdot g$ is added. Next, the binarity proof ($\binQuoCom$) and the sumcheck proof ($\apkQuoCom$) are modified to take the randomization into account.
This makes all proof components that depend on the bit vector uniformly distributed.
The exact prover computations are shown in lines \ref{line:apkChangesStart}--\ref{line:apkChangesEnd} in \cref{fig:ppolConstruct}.
The verifier hashes the bit commitments to compute $\gamma_3$ and then uses the following combined binarity and equivalence check as part of the \APKVerify{} routine (cf. \cref{sec:apkproof}):
\begin{align*}
e(\binCom - g, \hat{\binCom}) + \gamma_3 \cdot \left(e(\binCom, \hat{g}) - e(g, \hat{\binCom})\right) = e(\binQuoCom, [(\tau^n -1) \cdot \hat{g}]).
\end{align*}
The sumcheck verification step of \APKVerify{} remains unchanged since $\apkQuoCom$ already incorporates the randomization.

\end{document}